\documentclass[journal]{IEEEtran}
\usepackage{amssymb}
\usepackage{amsmath}
\usepackage{amsfonts}
\usepackage{graphicx}
\usepackage{algorithm}
\usepackage{algorithmic}
\usepackage{cite}
\usepackage{epstopdf}
\usepackage{stfloats}
\usepackage{color}
\usepackage{url}
\newtheorem{theorem}{\textbf{Theorem}}

\newtheorem{corollary}{\textbf{Corollary}}

\newtheorem{remark}{\textbf{Remark}}

\begin{document}
\bibliographystyle{IEEEtran}

\title{Performance Analysis for Training-Based Multi-Pair Two-Way Full-Duplex Relaying with Massive Antennas}
\author{Zhanzhan Zhang$^{*}$, Zhiyong Chen$^{*}$, Manyuan Shen$^{*}$, Bin Xia$^{*}$, Weiliang Xie$^{\dag}$, and Yong Zhao$^{\dag}$\\
${^*}$Department of Electronic Engineering, Shanghai Jiao Tong University, Shanghai, P. R. China\\
$^\dag$China Telecom Corporation Limited Technology Innovation Center\\
Email: \{mingzhanzhang, zhiyongchen, myshen, bxia\}@sjtu.edu.cn, \{xiewl, zhaoyong\}@ctbri.com.cn}
\maketitle

\begin{abstract}
This paper considers a multi-pair two-way amplify-and-forward relaying system, where multiple pairs of full-duplex users are served via a full-duplex relay with massive antennas, and the relay adopts maximum-ratio combining/maximum-ratio transmission (MRC/MRT) processing. The orthogonal pilot scheme and the least square method are firstly exploited to estimate the channel state information (CSI). When the number of relay antennas is finite, we derive an approximate sum rate expression which is shown to be a good predictor of the ergodic sum rate, especially in large number of antennas. Then the corresponding achievable rate expression is obtained by adopting another pilot scheme which estimates the composite CSI for each user pair to reduce the pilot overhead of channel estimation. We analyze the achievable rates of the two pilot schemes and then show the relative merits of the two methods. Furthermore, power allocation strategies for users and the relay are proposed based on sum rate maximization and max-min fairness criterion, respectively. Finally, numerical results verify the accuracy of the analytical results and show the performance gains achieved by the proposed power allocation.
\end{abstract}


\section{Introduction}
Massive multiple-input multiple-output (MIMO), an emerging technology which employs a few hundreds even thousands of antennas, is recently a very hot research topic in wireless communications \cite{massiveMIMOnextg}. By providing great array and spatial multiplexing gains, massive MIMO systems can be of higher spectral and energy efficiencies than conventional MIMO. Besides, the simplest linear precoders and detectors (such as maximum-ratio combining/maximum-ratio transmission (MRC/MRT)) can achieve optimal performance as the nonlinears \cite{massiveMIMOhowManyAntennas}. Therefore, massive MIMO is widely regarded as one of the cornerstone technologies for next-generation mobile networks. Furthermore, the related issues in terms of the implementation of massive MIMO technologies, such as the channel state information (CSI) acquisition and beamforming techniques, have been being proposed and discussed in recent 3GPP meetings \cite{massiveMIMO3GPPmeeting}.

On the other hand, full-duplex (FD) systems have attracted significant interest \cite{fdsystems}, due to the provided double spectral efficiency (SE) of traditional half-duplex (HD) systems.
However, by receiving and transmitting simultaneously on the same channel, FD systems suffer from a great
drawback of the inherent loop interference (LI) due to the signal leakage from the FD node output to input.

To suppress loop interference, many researches have already been done \cite{LIC10Stanford, LIC11Stanford, MitigationofLIfdMIMO, linearFDMIMORelay14C, LCdesignFDMIMO14TWC, Experiment12TWC}.
LI suppression approaches can be categorized as passive cancellation and active cancellation and the active cancellation further includes analog cancellation and digital cancellation.
For example, \cite{LIC10Stanford} showed that LI can be reduced to within a few dB of the noise floor by combining passive and active cancellations.
In \cite{LIC11Stanford}, the authors proposed the signal inversion and adaptive cancellation, which support wideband and high power systems, thus making it possible to build FD 802.11n devices.
In addition, \cite{MitigationofLIfdMIMO} extended the cancellation for single channel case to the FD MIMO relay case and proposed new spatial suppression techniques, such as minimum mean square error (MMSE) filtering.
The authors in \cite{linearFDMIMORelay14C} also studied the spatial processing techniques for a FD MIMO relay, and indicated that LI suppression is preferable to pre-cancellation at the relay transmitter. Then the joint precoding/decoding design with low complexity to mitigate LI in spatial domain for FD MIMO relaying was proposed in \cite{LCdesignFDMIMO14TWC}. In \cite{Experiment12TWC}, it was shown that the combination of digital and analog cancellation can sometimes increase the LI. Besides, it has been reported in \cite{Experiment12TWC, FDcancel110dB, FDinRelayand100dBsuppress} that 70-110 dB overall suppression of the LI can be realized. In a word, recent achievements in radio frequency (RF)/circuit design have made it feasible to perform full-duplex in certain scenarios. On the other hand, in the 3GPP process, it has been proposed by Huawei, NTT DOCOMO, etc., that new radio (NR) access technology should support of FD in the future in a forward compatible way \cite{massiveMIMO3GPPmeeting}.
Furthermore, in a recent paper \cite{massivemimo1}, the authors utilized the large-scale antennas to eliminate the loop interference due to the large array gain, and this discovery promotes the joint consideration of  massive MIMO and full-duplex in subsequent analysis. Nevertheless, the work considered the one-way relay with massive antennas, and similar research in such systems for two-way channels is barely addressed.

Inspired by both \textit{ad hoc} and infrastructure-based (e.g., cellular and WiFi) networks, two-hop wireless relaying is the most possible use case which can benefit from the FD operation \cite{inBandFD14JSAC, FDinRelayand100dBsuppress}, since in wireless relaying, the data traffic is inherently symmetric as far as the relay always froward the received information, and this character can efficiently utilize the FD ability of doubling the SE. Recently, FD wireless relaying has been discussed and included in the 3GPP standard \cite{FDrelayin3GPP}. In the experimental aspect, a FD MIMO relay for LTE-A has been studied in lab and proved to be technologically feasible \cite{FDinRelayand100dBsuppress}.
Then, in this paper, it is straightforward to generalize the FD wireless relaying model to the multi-pair two-way FD massive MIMO relay system model, by considering two-way relaying to more efficiently utilize the time/frequency resource. And either a mobile terminal or a base station can act as the FD relay.

Over the recent years, much progress has been made on  two-way  or one-way relaying systems.
For example, the authors in \cite{ManavTWAFOSTBC} studied the performance of a two-way amplify-and-forward (AF) MIMO relay system based on orthogonal space-time block codes (OSTBCs). In \cite{ManavTWSatelliteRelay}, a differential modulation based two-way relaying protocol was proposed for two-way AF satellite relaying communication.
In \cite{fdMIMOrelay}, the joint beamforming optimization and power control were investigated  for a two-way FD MIMO relay system. However, \cite{ManavTWAFOSTBC,ManavTWSatelliteRelay,fdMIMOrelay} all considered the traditional MIMO with a small number of antennas at the relay.
Besides, the power efficiency of a multi-pair AF relaying model with massive MIMO was investigated in \cite{HASura13ICC}, and it was shown that massive MIMO could greatly improve the power efficiency while maintaining a given quality-of-service. Nonetheless, it only considered the one-way HD relaying.
In addition, \cite{massivemimo2} studied the spectral efficiency and energy efficiency for a multi-pair two-way massive MIMO relay system, but only the case of infinite number of relay antennas was considered in \cite{massivemimo2}. Moreover,  \cite{ergodicRate} discussed the achievable ergodic rate with a finite number of relay antennas for the same system as \cite{massivemimo2}, however, both \cite{massivemimo2} and \cite{ergodicRate} dealt with the HD relays.

In this paper, we model a multi-pair two-way full-duplex AF relay system where the relay has a large-scale antenna array, in the presence of inter-user interference, and the MRC/MRT technique is considered. For massive MIMO systems, it is a big challenge to acquire the CSI.
In our system model, the relay needs to acquire the global CSI to perform MRC/MRT processing. The model involves both the uplink channels (from users to the receive antenna array of the relay) and the downlink channles (from the transmit antennas of the relay to users). In general, the estimations of channels are obtained by transmitting pilot signals. Since the frequency-division duplex (FDD) scheme, where users estimate the downlink CSI based on the pilot signals transmitted by the relay and feedback them to the relay, is prohibitive in massive MIMO relay networks \cite{zhengzhengxiang}, we consider the time division duplex (TDD) system where users transmit pilot signals to both the transmit and receive antennas of the relay, then the CSI estimated by the relay transmit antennas is considered as the downlink CSI based on channel reciprocity.

The contributions of this work are summarized as follows.
\begin{itemize}
  \item We derive a lower bound and an approximate expression for the ergodic sum rate with a finite and large number of relay antennas based on the statistical CSI, and the results are obtained by utilizing the orthogonal pilot scheme and least square (LS) channel estimation. The approximate sum rate expression is demonstrated to be a tight approximation to the ergodic sum rate. It is also shown that the sum rate can be increased significantly by adding the relay antenna number.
  \item We also derive the achievable rate expression by employing another pilot transmission scheme, which estimates the composite channel for each user pair. We present the comprehensive theoretical analysis on the achievable rates of the two pilot schemes and provide the valuable insights to show the relationship between the two methods.
  \item We derive the power allocation for maximizing the achievable sum rate and the minimum signal-to-interference-plus-noise ratio (SINR) of all users, respectively.
      Furthermore, we present the comparison of our scheme with other schemes and demonstrate that when the relay antenna number is very large, our scheme performs better than the corresponding one-way FD relaying scheme as well as the two-way HD relaying scheme.
\end{itemize}

The remainder of this paper is organized as follows. In Section \ref{systemModel}, the system model of the multi-pair massive MIMO two-way FD relay channel is described. In Section \ref{rateWithICE}, we derive an approximate sum rate expression based on conventional pilot scheme and LS channel estimation, when the relay antenna number is finite. In Section \ref{rateWithCCE}, we consider another pilot scheme and the corresponding achievable rate expression is also obtained. Section \ref{performanceEvaluation} compares the two pilot schemes and addresses the problem of power allocation.
Numerical results are provided in Section \ref{numericalResults}. Finally, Section \ref{conclusion} concludes the paper.

\textit{\textbf{Notations:}} Boldface uppercase and boldface lowercase letters denote  matrices and column vectors, respectively. $\mathbb{E}\{\cdot\}$, $\left\|\cdot\right\|_2$, $\mathrm{Tr}(\cdot)$, $(\cdot)^H$, $(\cdot)^T$, $(\cdot)^*$ stand for the expectation, Euclidean norm, the trace of a square matrix, the conjugate transpose, the transpose and the conjugate of a matrix, respectively. $\mathcal{CN}(\mathbf{x}, \mathbf{\Sigma})$ represents the distribution of a circularly symmetric complex Gaussian vector with mean vector  $\mathbf{x}$ and covariance matrix $\mathbf{\Sigma}$. $\textbf{I}_\mathrm{N}$ denotes an $N\times N$ identity matrix.
\section{System Model}\label{systemModel}
As shown in Fig. \ref{system model}, we consider a multi-pair two-way relaying network, where $K$ ($K\geqslant 2$) user pairs try to exchange information within pair through a relay ($R$) which operates in AF protocol. Let ($S_k,S_{k'}$) denote one source pair, $(k,k')=(2m-1,2m)$ or $(2m,2m-1)$, $m=1,2,\cdots,K$. Besides, all the user equipments and the relay operate in the FD mode, so that all nodes suffer from self-LI due to the simultaneous transmission and reception. Assume that each FD user has one FD antenna \cite{TWC09oneFDantenna, SPL12oneFDantenna, HongyuCuioneFDantenna}, and the FD relay is equipped with $N_r$ receive antennas and $N_t$ transmit antennas \footnote{We assume that the FD antenna at each user can be used for transmission and reception simultaneously. Based on this assumption and the orthogonality of the pilot sequences, the length of pilot symbols ($\tau$) consumed in the channel estimation stage would be at least the number of users ($2K$), which is half of that ($4K$) under the scenario where each user has two antennas in which one for transmission and the other for reception. In addition, the transmit antennas and receive antennas at the FD relay are separated,  then $N_t$ receiving RF chains, for receiving pilot signals during the channel estimation stage, are needed in the transmit antenna set apart from $N_t$ transmitting RF chains, while only $N_r$ receiving RF chains are needed in the receive antenna set.} \cite{massivemimo1, globecom05twoFDantenna}, and let $\kappa=N_t/N_r$.
We consider the scenario where the $K$ users with odd subscripts ($S_{2m-1}$) stay in one area and the other $K$ users with even subscripts ($S_{2m}$) stay in another area, thus direct links between $S_{k}$ and $S_{k'}$ do not exist due to the high path loss and shadow fading, while one user can inevitably receive signals from nearby users in the same area due to FD operation, and we regard this interference as inter-user interference.
In addition, we adopt the linear precoder and detector MRC/MRT at the relay in this paper, which is a common technique in the massive MIMO system.

%\subsection{Overall Description}

\begin{figure}
  \centering
  \includegraphics[width=3.5in]{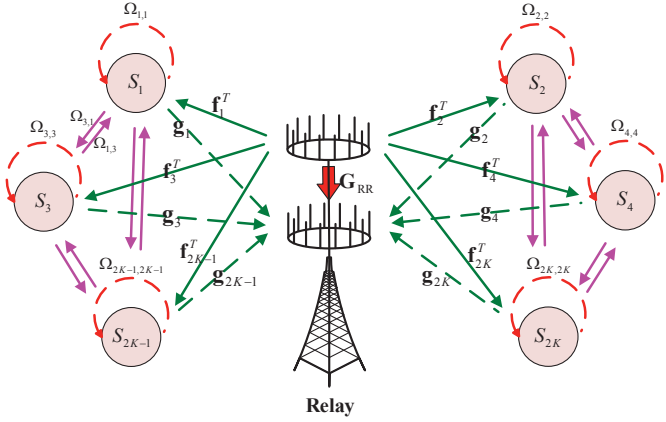}
  \caption{Multi-pair two-way full-duplex massive MIMO relay system in the presence of inter-user interference.}\label{system model}
\end{figure}

\newcounter{myEquationCounter}
\begin{figure*}[!t] % 加 [bp] 的话是放在页面底端！！！
% ensure that we have normalsize text
\normalsize
% Store the current equation number.

% Set the equation number to one less than the one
% desired for the first equation here.
% The value here will have to changed if equations
% are added or removed prior to the place these
% equations are referenced in the main text.
%%\setcounter{equation}{1}
\setcounter{myEquationCounter}{\value{equation}}
\setcounter{equation}{6}
\begin{equation}\label{eq8}
\tilde{y}_k=\underbrace{\alpha\sqrt{P_S}\mathbf{f}_\mathrm{k}^T\mathbf{W}\mathbf{g}_{\mathrm{k}'}x_{k'} }_{\text{desired signal}}
+\underbrace{\alpha\sqrt{P_S}\sum_{j=1\atop j\neq k,k'}^{2K}{\mathbf{f}_\mathrm{k}^T\mathbf{W}\mathbf{g}_{\mathrm{j}}x_{j}}}_{\text{inter-pair interferences}}
+\underbrace{\alpha\mathbf{f}^T_\mathrm{k}\mathbf{WG}_{\mathrm{RR}}\mathbf{\tilde{x}}_\mathrm{R} + \alpha\mathbf{f}^T_\mathrm{k}\mathbf{Wz}_\mathrm{R}}_{\text{LI and noise from the relay}}
+\underbrace{\sqrt{P_S}\sum_{ i\in\mathrm{U}_{k} }{\Omega_{k,i}x_{i}}  }_{\text{inter-user interferences and self-LI}}
+\underbrace{z_k}_{\text{noise} }.
\end{equation}
% Restore the current equation number.
%%\setcounter{equation}{\value{mytempeqncnt}}
% IEEE uses as a separator
\hrulefill
% The spacer can be tweaked to stop underfull vboxes.
\vspace*{0pt}
\end{figure*}

Before further description, we assume that some traditional loop interference cancellation (LIC) techniques have been executed at the users and the relay in this paper, such as applying RF attenuation, time-domain suppression and/or spatial cancellation techniques \cite{li1,SPL12oneFDantenna,MitigationofLIfdMIMO}. Then the residual LI channels can be modeled as Rayleigh fading distribution \cite{massivemimo1,MitigationofLIfdMIMO}. Furthermore, the residual LIs due to the imperfection of LIC methods are assumed to be additional Gaussian noise variables \cite{HongyuCuioneFDantenna, MitigationofLIfdMIMO, LImodelJSAC, LInoiseModelAndLICmethodsurvey}. This assumption will be the worst-case scenario regarding the achievable data rate if the residual LI is not Gaussian \cite{LImodelJSAC, GaussianNoiseWorstCase}.

\subsection{Signal Model}
At time instant $n$, $S_k$ ($k=1,2, \cdots, 2K$) transmits the signal $\sqrt{P_S}x_k(n)$ to the relay, and at the same time, the relay broadcasts the signal $\mathbf{x}_\mathrm{R}(n)\in \mathbb{C}^{N_t\times1}$ to all source nodes. Here, we consider that each user has the same transmit power $P_S$ and $\mathbb{E}\left\{|x_k(n)|^2\right\}=1$. The transmit power of the relay is restricted by $P_R$, so we have $P_R=\mathbb{E}\left\{\mathrm{Tr}\left[\mathbf{x}_\mathrm{R}(n)\mathbf{x}_\mathrm{R}^H(n)\right]\right\}$. Therefore the received signals at the relay and the source node $S_k$ are, respectively
\setcounter{equation}{0}
\begin{equation}
\mathbf{y}_\mathrm{R}(n)=\sqrt{P_S}\mathbf{G}\mathbf{x}(n)+\mathbf{G}_{\mathrm{RR}}\mathbf{x}_\mathrm{R}(n)+\mathbf{z}_\mathrm{R}(n), \label{eq1}
\end{equation}
\begin{equation}
y_k(n)=\mathbf{f}_\mathrm{k}^T\mathbf{x}_\mathrm{R}(n)+\sum_{ i\in\mathrm{U}_{k} }{\Omega_{k,i}\sqrt{P_S}x_{i}(n)}+z_k(n),\label{eq2}
\end{equation}
where the set $\mathrm{U}_k=\{1,3,\cdots,2K-1\}$ if $k$ is an odd number or $\mathrm{U}_k=\{2,4,\cdots,2K\}$ otherwise, and $\mathbf{x}(n)=[x_1(n),x_2(n),\cdots,x_{2K}(n)]^T$. Let us define $\mathbf{G}=[\mathbf{g}_1,\mathbf{g}_2,\cdots,\mathbf{g}_{2\mathrm{K}}]$, where $\mathbf{g}_\mathrm{k}\in\mathbb{C}^{N_r\times1}$ denotes the uplink channels between the antenna of $S_k$ and the receive antenna array of the relay. Also we define $\mathbf{F}=[\mathbf{f}_1,\mathbf{f}_2,\cdots,\mathbf{f}_{2\mathrm{K}}]$, where $\mathbf{f}_\mathrm{k}^T\in\mathbb{C}^{1\times N_t}$ denotes the downlink channels from the transmit antenna array of $R$ to the antenna of $S_k$.
$\mathbf{G}$ and $\mathbf{F}$ are assumed to obey the independent identically distributed (i.i.d.) Rayleigh fading and therefore
$\mathbf{g}_\mathrm{k} \thicksim \mathcal{CN}(\mathbf{0},\beta_{uk}\mathbf{I}_\mathrm{N_r})$ and $\mathbf{f}_\mathrm{k} \thicksim \mathcal{CN}(\mathbf{0},\beta_{dk}\mathbf{I}_\mathrm{N_t})$.
Hence, $\mathbf{G}$ and $\mathbf{F}$ can be expressed as $\mathbf{G}=\mathbf{H}_\mathrm{u}\mathbf{D}^{1/2}_\mathrm{u}$ and $\mathbf{F}=\mathbf{H}_\mathrm{d}\mathbf{D}^{1/2}_\mathrm{d}$, respectively, where $\mathbf{H}_\mathrm{u}$ and $\mathbf{H}_\mathrm{d}$ denote the small-scale fading with i.i.d. $\mathcal{CN}(0,1)$ random entries,
$\mathbf{D}_\mathrm{u}$ and $\mathbf{D}_\mathrm{d}$ are diagonal matrices representing the large-scale fading, and the $k$-th diagonal elements of $\mathbf{D}_\mathrm{u}$ and $\mathbf{D}_\mathrm{d}$ are denoted as $\beta_{uk}$ and $\beta_{dk}$, respectively.
In addition, $\mathbf{G}_{\mathrm{RR}}\in\mathbb{C}^{N_r\times N_t}$ and $\Omega_{k,k}$ denote the self-LI channel coefficients at the relay $R$ and user $S_k$ respectively, and the entries of $\mathbf{G}_{\mathrm{RR}}$ and $\Omega_{k,k}$ are i.i.d. $\mathcal{CN}(0,\sigma^2_{LI})$ and $\mathcal{CN}(0,\sigma^2_{k,k})$ random variables, respectively.
$\Omega_{k,i}$ ($i\in\mathrm{U}_{k},i\neq k$) represents the inter-user interference channel coefficient from $S_i$ to $S_k$, and assume $\Omega_{k,i} \thicksim \mathcal{CN}(0,\sigma^2_{k,i})$ \cite{InterUserInterferenceModel}. $\mathbf{z}_\mathrm{R}(n) \thicksim \mathcal{CN}(\mathbf{0},\sigma^2_{nr}\mathbf{I}_\mathrm{N_r})$ is an additive white Gaussian noise (AWGN) vector at the relay and $z_k(n) \thicksim \mathcal{CN}(0,\sigma^2_n)$ means AWGN at $S_k$.

Let $\tau_d$ ($\tau_d\geqslant 1$) denote the processing delay of the relay. At time instant $n$ ($n>\tau_d$), the relay $R$ amplifies the previously received signal $\mathbf{y}_\mathrm{R}(n-\tau_d)$ and broadcasts it to the sources. We thus have
\begin{equation}\label{eq3}
\mathbf{x}_\mathrm{R}(n)=\alpha\mathbf{W}\mathbf{y}_\mathrm{R}(n-\tau_d),
\end{equation}
where $\mathbf{W}\in\mathbb{C}^{N_t\times N_r}$ is the relay processing matrix, and $\alpha$ denotes a power constraint factor at the relay.

Due to the processing delay of the relay, we assume that the transmitted signal $\mathbf{x}_\mathrm{R}(n)$ of the relay is uncorrelated with the received signal $\mathbf{y}_{\mathrm{R}}(n)$ \cite{massivemimo1, relayDelayTWC2011, relayDelayTWC2013}. In addition, after performing some LIC techniques, let $\mathbf{G}_{\mathrm{RR}}\mathbf{\tilde{x}}_\mathrm{R}(n)$ represent the residual LI at the relay. And because the amount of LI is mainly decided by the transmit power $P_R$, we have $\mathbf{G}_{\mathrm{RR}}\mathbf{\tilde{x}}_\mathrm{R}(n) \sim \mathcal{CN}(\mathbf{0}, P_R\sigma_{LI}^2\mathbf{I}_{\mathrm{N_r}})$ according to the previous assumption of the residual LI.
Then, substituting (\ref{eq1}) into (\ref{eq3}) and owing to the power constraint of the relay, i.e. $\mathbb{E}\left\{\mathrm{Tr}\left[\mathbf{x}_\mathrm{R}(n)\mathbf{x}_\mathrm{R}^H(n)\right]\right\}{=}P_R$,
we have\footnote{We consider that the relay adopts the statistical CSI instead of instantaneous CSI to derive $\alpha$, known as the ``fixed gain relay'' \cite{fixedGainRelay}. Then the relay has a long-term power constraint $P_R{=}\mathbb{E}\left\{\mathrm{Tr}\left[\mathbf{x}_\mathrm{R}(n)\mathbf{x}_\mathrm{R}^H(n)\right]\right\}$ where the expectation is taken over the channel realizations as well as the signal and the noise. Note that the ``fixed gain relay'' has lower complexity and is easier to deploy than the ``variable gain relay'' using  instantaneous CSI to derive $\alpha$.}

\setcounter{equation}{3}
\begin{align}\label{eq4}
\alpha=\sqrt{ \frac{P_R}{ P_S\cdot \Delta_1 +(P_R\sigma_{LI}^2+\sigma^2_{nr})\cdot \Delta_2 } },
\end{align}
in which
\begin{align}
&\Delta_1 = \mathbb{E}\left[ \mathrm{Tr}\left( \mathbf{WGG}^H\mathbf{W}^H \right) \right], \label{eq5}\\ %=\mathrm{Tr}(\mathbf{WGG}^H\mathbf{W}^H)
&\Delta_2 = \mathbb{E}\left[ \mathrm{Tr}\left( \mathbf{WW}^H \right) \right]. \label{eq6}
\end{align}

%\begin{align}\label{eq14}
%\tilde{y}_k&=\alpha\mathbf{f}_\mathrm{k}^T\mathbf{W}\mathbf{g}_{\mathrm{k}'}\sqrt{P_S}x_{k'} +z_k \nonumber\\
%&+\alpha\mathbf{f}_\mathrm{k}^T\mathbf{W}\sum_{j=1\atop j\neq k,k'}^{2K}{\mathbf{g}_{\mathrm{j}}\sqrt{P_S}x_{j}}+\sum_{i,k\in\mathrm{U}_{k}}{\Omega_{k,i}\sqrt{P_S}x_{i}} \nonumber\\
%&+\alpha\mathbf{f}^T_\mathrm{k}\mathbf{WG}_{\mathrm{RR}}\mathbf{\tilde{x}}_\mathrm{R}+\alpha\mathbf{f}^T_\mathrm{k}\mathbf{Wz}_\mathrm{R}.
%\end{align}

Then, substituting (\ref{eq1}) and (\ref{eq3}) into (\ref{eq2}), we can get the received signal at $S_k$ in detail as represented by (\ref{eq8}) (see top of next page\footnote{Note that the self-interference ($\alpha\sqrt{P_S}\mathbf{f}_\mathrm{k}^T\mathbf{W}\mathbf{g}_{\mathrm{k}}x_{k}$) which stems from $S_k$ and is amplified and forwarded to itself by the relay due to two-way relaying  is not included in (\ref{eq8}) after applying the self-interference cancellation (SIC) technique. We will introduce the SIC briefly in the following derivation.}),
where the time labels are omitted, we also omit the time labels hereinafter for convenience.
%Note that the self-interference ($\alpha\mathbf{f}_\mathrm{k}^T\mathbf{W}\mathbf{g}_{\mathrm{k}}\sqrt{P_S}x_{k}$) is neglected due to the self-interference cancellation (SIC) technique \cite{ergodicRate}.
It is seen that the first term of the right hand side of (\ref{eq8}) is the desired signal.
The second term denotes the inter-pair interferences which are transmitted by other source pairs and then are amplified and forwarded to $S_k$ by the relay. The third and fourth terms indicate that the residual LI due to the FD operation of the relay and the noise at the relay are also forwarded to the user by the relay, respectively. The fifth term consists of the inter-user interferences ($\sqrt{P_S}\sum_{i\in \mathrm{U}_k,i\neq k}{\Omega_{k,i}x_i}$) which are caused by nearby users and the self-LI ($\sqrt{P_S}\Omega_{k,k}x_k$) arising from the FD operation of the user itself. And the last term is the local noise.

\section{Achievable Rate Analysis with individual channel estimation}\label{rateWithICE}
In this section, we derive the achievable rate for the multi-pair two-way FD relay system, when the number of relay antennas is finite.

\begin{figure*}[!t] % 加 [bp] 的话是放在页面底端！！！
% ensure that we have normalsize text
\normalsize
\setcounter{equation}{17}
\begin{equation}\label{eq14}
\gamma_k {=} \frac{ P_S \left|  \mathbb{E}\left\{ \mathbf{f}^T_\mathrm{k}\mathbf{W}\mathbf{g}_{\mathrm{k}'} \right\} \right|^2}
{P_S \mathbb{V}\text{ar}( \mathbf{f}^T_\mathrm{k}\mathbf{W}\mathbf{g}_{\mathrm{k}'} )
{+} P_S\sum\limits_{j=1\atop j\neq k,k'}^{2K}{ \mathbb{E}\left\{\left|\mathbf{f}^T_\mathrm{k}\mathbf{W}\mathbf{g}_{\mathrm{j}}\right|^2 \right\} }
{+} \left(P_R\sigma_{LI}^2 {+} \sigma^2_{nr}\right) \mathbb{E}\left\{ \left\|\mathbf{f}^T_\mathrm{k}\mathbf{W}\right\|_2^2 \right\}
{+} \frac{P_S\sum\limits_{i\in\mathrm{U}_{k}}{\sigma_{k,i}^2} {+} \sigma^2_n}{\alpha^2} }.
\end{equation}
\hrulefill
% The spacer can be tweaked to stop underfull vboxes.
\vspace*{0.0pt}
\end{figure*}

\subsection{Individual Channel Estimation (ICE)}
In this paper, we consider the flat block-fading channel, i.e. the channels during a block keep constant and vary independently across different blocks. The coherence interval (in symbols) of a block is denoted by $T_c$.
Suppose that $\tau$ symbols of the coherence interval $T_c$ are consumed in the pilot transmission phase. All users transmit deterministic pilot sequences ($\sqrt{\tau P_p} \mathbf{\phi}_k\in \mathbb{C}^{1\times \tau}$, $k=1,2,\cdots,2K$) to the relay simultaneously, where $\mathbf{\phi}_k\phi_k^H=1$ and $P_p$ denotes the transmit power of each pilot symbol. Then the received pilot signals at the receive and transmit antenna arrays of the relay are represented by, respectively
\setcounter{equation}{7}
\begin{equation}\label{iCE1}
\mathbf{Y}_{\mathrm{rp}} = \sqrt{\tau P_p} \mathbf{G} \mathbf{\Phi} + \mathbf{Z}_{\mathrm{rp}},
\end{equation}
\begin{equation}\label{iCE2}
\mathbf{Y}_{\mathrm{tp}} = \sqrt{\tau P_p} \mathbf{F} \mathbf{\Phi} + \mathbf{Z}_{\mathrm{tp}},
\end{equation}
where $\mathbf{\Phi}=[\phi_1^T, \phi_2^T, \cdots, \phi_{2K}^T]^T \in \mathbb{C}^{2K\times \tau}$ is the transmitted pilot signal matrix, $\mathbf{Z}_{\mathrm{rp}}\in \mathbb{C}^{N_r\times \tau}$ and $\mathbf{Z}_{\mathrm{tp}}\in \mathbb{C}^{N_t\times \tau}$ denote the AWGN matrices with their elements are all $\mathcal{CN}(0, \sigma_{nr}^2)$ random variables. The goal of channel estimation is to obtain individual CSI for each user, thus all pilot sequences need to be orthogonal to each other, i.e. $\mathbf{\Phi}\mathbf{\Phi}^H=\mathbf{I}_{\mathrm{2K}}$, which requires $\tau\geqslant2K$.

In this paper, the popular LS channel estimation\footnote{ The reason we adopt the LS estimation method here is that it has the lowest complexity than other estimation methods and the main focus of this paper is to study and compare the impact of different pilot schemes on the system performance. Although its estimation error is a little larger than other methods, such as MMSE channel estimation, the estimation error will be cancelled out by the large array gain offered by massive MIMO. Thus we consider the popular LS approach.} \cite{LSchannelEstimation} is applied at the relay, and the LS estimations of the matrices $\mathbf{G}$ and $\mathbf{F}$ are given by
\begin{equation}\label{iCE3}
\mathbf{\hat{G}} = \frac{1}{\sqrt{\tau P_p}} \mathbf{Y}_{\mathrm{rp}} \mathbf{\Phi}^H = \mathbf{G} + \mathbf{Z}_\mathrm{r},
\end{equation}
\begin{equation}\label{iCE4}
\mathbf{\hat{F}} = \frac{1}{\sqrt{\tau P_p}} \mathbf{Y}_{\mathrm{tp}} \mathbf{\Phi}^H = \mathbf{F} + \mathbf{Z}_\mathrm{t},
\end{equation}
respectively, where $\mathbf{Z}_\mathrm{r} {=} \frac{1}{\sqrt{\tau P_p}} \mathbf{Z}_{\mathrm{rp}} \mathbf{\Phi}^H$ and $\mathbf{Z}_\mathrm{t} {=} \frac{1}{\sqrt{\tau P_p}} \mathbf{Z}_{\mathrm{tp}} \mathbf{\Phi}^H$ indicate the relevant estimation error matrices and their entries have zero means and  variances of $\frac{\sigma_{nr}^2}{\tau P_p}$. Apparently, the actual channel matrices $\mathbf{G}$ and $\mathbf{F}$ are independent with the error matrices $\mathbf{Z}_{\mathrm{r}}$ and $\mathbf{Z}_{\mathrm{t}}$, hence the large-scale fading matrices are estimated as
\begin{equation}\label{iCE5}
\mathbf{\hat{D}}_{\mathrm{u}}= \mathbf{D}_{\mathrm{u}} + \frac{\sigma_{nr}^2}{\tau P_p}\mathbf{I}_{\mathrm{2K}},
\end{equation}
\begin{equation}\label{iCE6}
\mathbf{\hat{D}}_{\mathrm{d}}= \mathbf{D}_{\mathrm{d}} + \frac{\sigma_{nr}^2}{\tau P_p}\mathbf{I}_{\mathrm{2K}},
\end{equation}
where the $i$-th diagonal elements of $\mathbf{\hat{D}}_\mathrm{u}$ and $\mathbf{\hat{D}}_\mathrm{d}$ are denoted by $\hat{\beta}_{ui}$ and $\hat{\beta}_{di}$, respectively.

\subsection{Achievable Rate: A Lower Bound}
From (\ref{eq8}), we can obtain the ergodic sum rate of the multi-pair two-way FD relay system with massive MIMO processing  as represented by
\begin{equation}\label{eq10}
C = \frac{T_c-\tau}{T_c} \mathbb{E}\left\{\sum_{k=1}^{2K}\mathrm{log_2}(1+\text{SINR}_k)\right\},
\end{equation}
where $\mathrm{SINR}_k$ denotes the received instantaneous   SINR at the user node $S_k$.

However, it is extremely difficult to derive a closed-form expression of the system capacity from (\ref{eq10}).
%Besides, the instantaneous channel state information (CSI) is needed for obtaining the ergodic capacity, yet the CSI acquisition of massive MIMO systems is very challenging. %which requires that the channel estimation should be very well.
Therefore, instead of calculating (\ref{eq10}) directly, we refer to the technique from \cite{howMuchTraining} which is widely used in the regime of massive MIMO \cite{massiveMIMOhowManyAntennas, massivemimo1, lowerBoundR11TWC, lowerBoundR15TIT, lowerBoundR15TWC}.
This technique utilizes the statistical channels to detect the received signals. With this technique, the received signal expression (\ref{eq8}) can be rewritten as
\begin{equation}\label{eq11}
\tilde{y}_k= \alpha\sqrt{P_S}\mathbb{E}\{\mathbf{f}_\mathrm{k}^T\mathbf{W}\mathbf{g}_{\mathrm{k}'}\} x_{k'} + \tilde{z}_k,
\end{equation}
where $\tilde{z}_k$ is defined as  the effective noise at $S_k$, and $\tilde{z}_k$ is given by
\begin{align}\label{eq12}
\tilde{z}_k &\triangleq \alpha\sqrt{P_S}\left(\mathbf{f}_\mathrm{k}^T\mathbf{W}\mathbf{g}_{\mathrm{k}'} - \mathbb{E}\{\mathbf{f}_\mathrm{k}^T\mathbf{W}\mathbf{g}_{\mathrm{k}'}\}\right) x_{k'} \nonumber\\
&+ \alpha\sqrt{P_S}\sum_{j=1\atop j\neq k,k'}^{2K}{\mathbf{f}_\mathrm{k}^T\mathbf{W}\mathbf{g}_{\mathrm{j}}x_{j}}
+ \alpha\mathbf{f}^T_\mathrm{k}\mathbf{WG}_{\mathrm{RR}}\mathbf{\tilde{x}}_\mathrm{R}+\alpha\mathbf{f}^T_\mathrm{k}\mathbf{Wz}_\mathrm{R} \nonumber\\
&+ \sqrt{P_S}\sum_{i\in\mathrm{U}_{k}}{\Omega_{k,i}x_{i}}
+ z_k.
\end{align}

Fortunately, it is easy to verify that the expected desired signal ($\mathbb{E}\{\mathbf{f}_\mathrm{k}^T\mathbf{W}\mathbf{g}_{\mathrm{k}'}\} x_{k'}$) and the effective noise ($\tilde{z}_k$) are uncorrelated. Based on the Theorem 1 in \cite{howMuchTraining} which states that the worst case uncorrelated additive noise is independent Gaussian noise with the same variance in terms of the mutual information, we arrive at an achievable data rate of the system shown as
\begin{equation}\label{eq13}
R = \frac{T_c-\tau}{T_c}\sum_{k=1}^{2K}\mathrm{log_2}(1+\gamma_k).
\end{equation}
where the statistical SINR $\gamma_k$ is given by (\ref{eq14}) based on (\ref{eq11}) and (\ref{eq12}).

\begin{remark}\label{remarkLowerbound}
Since the worst case uncorrelated Gaussian noise property is used to derive $\gamma_k$, it is expected that the rate expression (\ref{eq13}) is a lower bound of the ergodic rate, i.e. ($R\leqslant C$). And it will be demonstrated via the numerical results that the performance gap between the lower bound and the achievable ergodic rate is very small, which verifies that the lower bound is a good predictor of the achievable rate.
\end{remark}

\subsection{An Approximate Rate Expression}
According to \cite{massivemimo2}, the MRC/MRT processing  matrix is given by
%\bigskip
\setcounter{equation}{18}
\begin{equation}\label{eq15}
\mathbf{W}=\mathbf{\hat{F}}^*\mathbf{T}\mathbf{\hat{G}}^H = (\mathbf{F} + \mathbf{Z}_\mathrm{t})^*\mathbf{T}(\mathbf{G} + \mathbf{Z}_\mathrm{r})^H,
\end{equation}
where $\mathbf{T}=\mathrm{diag}(\mathbf{T}_1,\mathbf{T}_2,\cdots,\mathbf{T}_\mathrm{K})$ is the diagonal permutation matrix indicating the exchange of information between each user pair, and $\mathbf{T}_\mathrm{m}=[0~1;1~0]$, for any $m=1,2,\cdots,K$.

Substituting (\ref{eq15}) into (\ref{eq5}) and (\ref{eq6}), we have
\begin{align}
&\Delta_1 = N_t\sum_{i=1}^{2K}{ \hat{\beta}_{di}\left(N_r^2\beta_{ui'}^2 + N_r\hat{\beta}_{ui'}\sum_{j=1}^{2K}{ \beta_{uj} } \right) }, \label{eq16}\\
&\Delta_2 = N_tN_r\sum_{i=1}^{2K}{ \hat{\beta}_{di}\hat{\beta}_{ui'} }. \label{eq17}
\end{align}
Equations (\ref{eq16}) and (\ref{eq17}) are proved in Appendix A.
Thus we can obtain the power constraint factor $\alpha$ by substituting (\ref{eq16}) and (\ref{eq17}) into (\ref{eq4}).

In the following theorem, we derive an approximate closed-form expression of the achievable lower bound given by (\ref{eq13}).
\begin{theorem}\label{theorem1}
With a fixed value of $\kappa$, when the number of relay antennas is finite and $N_r\gg 2K$, an approximate closed-form expression for the SINR of user $S_k$ under MRC/MRT processing is represented by
\begin{equation}\label{eq19}
\gamma_k
\approx \frac{N_t}{A_k + {\text{MP}}_k+{\text{LIR}}_k+{\text{NR}}_k+{\text{MU}}_k+{\text{AN}}_k},
\end{equation}
where
\begin{align}
&A_k = \kappa \frac{\hat{\beta}_{uk'}}{\beta_{uk'}} + \frac{\hat{\beta}_{dk}}{\beta_{dk}}, \label{AeffectDeviation}\\
&\text{MP}_k =
\sum_{j=1\atop j\neq k,k'}^{2K} { \left(\kappa\frac{\beta_{uj}\hat{\beta}_{uk'}}{\beta_{uk'}^2}+\frac{\hat{\beta}_{dj'}\beta_{uj}^2}{\beta_{dk}\beta_{uk'}^2}\right) }, \label{eq20} \\
&\text{LIR}_k = \frac{P_R\sigma_{LI}^2}{P_S}\frac{\kappa\hat{\beta}_{uk'}}{\beta_{uk'}^2}, \label{eq21} \\
&\text{NR}_k = \frac{\sigma_{nr}^2}{P_S}\frac{\kappa\hat{\beta}_{uk'}}{\beta_{uk'}^2}, \label{eq22}\\
&{\text{MU}}_k = \frac{1}{\beta_{dk}^2\beta_{uk'}^2}
\Delta_3 \sum_{ i\in\mathrm{U}_{k} }{\sigma_{k,i}^2} , \label{eq23}\\
&{\text{AN}}_k = \frac{\sigma_n^2}{P_S\beta_{dk}^2\beta_{uk'}^2} \Delta_3, \label{eq24}
\end{align}
and\footnote{Considering the power-scaling law: $P_S = E_S/N_r$ and $P_R = E_R/N_t$ where $E_S$ and $E_R$ are fixed, we have  $\Delta_3=\frac{P_S}{P_R}\sum_{i=1}^{2K}{\hat{\beta}_{di}\beta_{ui'}^2} + \frac{\sigma_{nr}^2}{N_rP_R} \sum_{i=1}^{2K}{\hat{\beta}_{di}\hat{\beta}_{ui'}}$.} $\Delta_3=\frac{P_S}{P_R}\sum_{i=1}^{2K}{\hat{\beta}_{di}\beta_{ui'}^2}$.
\end{theorem}
\begin{proof}
See Appendix B.
\end{proof}

Theorem 1 provides an approximate achievable rate expression when the number of relay antennas is large and finite. We observe that the small-scale fading is averaged out and the achievable rate is decided by the large-scale fading coefficients, which is the advantage of using the statistical channels for signal detection. On the other hand, since only the average effective channel $\mathbb{E}\{\mathbf{f}_\mathrm{k}^T\mathbf{W}\mathbf{g}_{\mathrm{k}'}\}$ is utilized for detection, there will be a deviation from the instantaneous channel, which is denoted by $A_k$. In addition, it is easy to discover that $\text{MP}_k$ represents the inter-pair interference; $\text{LIR}_k$ and $\text{NR}_k$ denote LI and noise from the relay, respectively; $\text{MU}_k$ signifies the inter-user interference and self-LI; $\text{AN}_k$ indicates the additive noise at $S_k$.
Furthermore, (\ref{eq19}) indicates that increasing the transmit antenna number of the relay can greatly enhance the sum rate, and approximately logarithmically in very large $N_t$.

Next, we investigate the best relation between $N_t$ and $K$ with which the sum rate will achieve its peak value.

For simplicity of analysis, we consider the case where all large-scale fading coefficients are normalized to be 1, i.e., $\mathbf{D}_\mathrm{u} = \mathbf{D}_\mathrm{d} = \mathbf{I}_{\mathrm{2K}}$. Without loss of generality, consider perfect CSI with no channel estimation error, $\sigma_{nr}^2 = \sigma_n^2$, $\sigma_{LI}^2 = \sigma_{k,i}^2$ ($i\in \mathrm{U}_k$ and $\forall k$), $P_R = KP_S$ and $N_t = N_r$. Then, based on (\ref{eq19}) $\sim$ (\ref{eq24}), the SINR for any user is given by
\begin{equation}\label{sinrNtK1}
\gamma_k = \frac{N_t}{aK - b}, \forall k,
\end{equation}
where $a = 3\sigma_{LI}^2 + 4$ and $b = 2-\frac{3\sigma_n^2}{P_S}$. Let $\xi_S = P_S/\sigma_n^2$, and obviously $\xi_S$ indicates the transmit signal-to-noise ratio (SNR) of users. Thereby, the lower bound in (17) is represented as
\begin{equation}\label{lowerBoundRes1}
R = \frac{T_c-\tau}{T_c} 2K\log_2\left(1+\frac{N_t}{aK-b}\right).
\end{equation}

By taking the first order derivative $R'(K)$ of $R$ with respect to $K$, and letting $R'(K) = 0$, we have
\begin{equation}\label{bestRelationNtK1}
\ln\left( 1+\frac{N_t}{aK-b} \right) = \frac{N_t aK}{(aK-b+N_t)(aK-b)},
\end{equation}
which shows the best relation between $N_t$ and $K$. The ``best relation'' means that the sum rate will achieve its peak value when the number of users $K$ satisfies  (\ref{bestRelationNtK1}) here. However, it's nontrivial to obtain some meaningful insights from (\ref{bestRelationNtK1}). Indeed, in massive MIMO case and when $N_t \gg aK-b$, we can obtain the following expression from (\ref{bestRelationNtK1})
\begin{equation}\label{approxiRelationNtK}
N_t \approx (aK-b)e^\frac{aK}{aK-b}.
\end{equation}
By differentiation with respect to $K$, we have
\begin{equation}\label{differentiationNtK1}
N_t'(K) = \frac{a(aK-2b)}{aK-b} e^{\frac{aK}{aK-b}}.
\end{equation}
Note that $aK-2b = 3\sigma_{LI}^2 K + 4(K-1) + 6/\xi_S > 0$ with $K \geqslant 1$, thus we have $N_t'(K)>0$. Therefore, for satisfying the best relation,  the  required transmit antenna number is increasing with respect to the optimal $K$.

Furthermore, (\ref{approxiRelationNtK}) implies that with fixed $N_t$, the sum rate will increase with the number of user pairs. But when the number of user pairs is larger than the optimal $K$ which satisfies the best relation, the sum rate will decline. And this insight will be verified by the simulation results in Fig. 2.

\begin{remark}\label{remarkICESIC}
About the self-interference
$\alpha\sqrt{P_S}\mathbf{f}_\mathrm{k}^T\mathbf{W}\mathbf{g}_{\mathrm{k}}x_{k}$, when $N_r$ and $N_t$ are very large and based on the law of large numbers (\textit{Lemma 1} in \cite{massivemimo2}), we have
\begin{align}\label{SIice1}
\mathbf{f}_\mathrm{k}^T\mathbf{W}\mathbf{g}_\mathrm{k}
& = \mathbf{f}_\mathrm{k}^T \left[ \sum_{i=1}^{2K}{\left(\mathbf{f}_\mathrm{i}^*+\mathbf{z}_{\mathrm{ti}}^*\right) \left(\mathbf{g}_\mathrm{i'}^H+\mathbf{z}_{\mathrm{ri'}}^H\right)}  \right] \mathbf{g}_\mathrm{k} \nonumber\\
& \approx \left\|\mathbf{f}_\mathrm{k}\right\|_2^2\mathbf{\hat{g}}_\mathrm{k'}^H\mathbf{g}_\mathrm{k}
+\mathbf{f}_\mathrm{k}^T\mathbf{\hat{f}}_\mathrm{k'}^*\left\|\mathbf{g}_\mathrm{k}\right\|_2^2,
\end{align}
where $\mathbf{z}_{\mathrm{ti}}$ and $\mathbf{z}_{\mathrm{ri}}$ are the $i$-th columns of $\mathbf{Z}_{\mathrm{t}}$ and $\mathbf{Z}_{\mathrm{r}}$, respectively.
We see that only the CSI of the user pair ($S_k$, $S_{k'}$) is required for $S_k$ to perform SIC when $N_r$ and $N_t$ are large. In addition, when $\kappa$ is fixed and $N_r\rightarrow\infty$, we get $\frac{\mathbf{f}_\mathrm{k}^T\mathbf{W}\mathbf{g}_\mathrm{k}}{N_tN_r}\rightarrow 0$, while $\frac{\mathbf{f}_\mathrm{k}^T\mathbf{W}\mathbf{g}_\mathrm{k'}}{N_tN_r}\rightarrow \beta_{dk}\beta_{uk'}$, thus SIC is needless when $N_r\rightarrow\infty$.
\end{remark}

\section{Achievable Rate Analysis with composite channel estimation}\label{rateWithCCE}
In the previous section, every user's CSI can be estimated by pilot-based channel estimation at the cost of at least $2K$ pilot symbols, and only ($T_c-2K$) symbols are left for payload transmission. When $T_c$ is small, the achievable data rate would be very little. Motivated by \cite{zhengzhengxiang} in which the scheme, where all users in a cell exploit the same pilot sequence and different cells use orthogonal pilot sequences, is proposed to eliminate the inter-cell interference, we are interested in investigating the performance for our system model when two users in each user pair employ the same pilot sequence and different user pairs adopt orthogonal pilot sequences. With this pilot scheme, the minimum pilot sequence length can be reduced to a half, i.e. only $K$ pilot symbols are required at least. As a result, the relay can only estimate the composite channels for each user pair instead of each user's CSI.

In addition, this pilot scheme was also employed in \cite{Xiaojunzheng}, where the performance  was evaluated for the multi-pair two-way relay system when the number of relay antennas went to infinity. However, \cite{Xiaojunzheng} only evaluated the performance when $T_c$ was little ($T_c=10$ therein), and the performance in the regime of large coherence interval is worth exploring. Besides, only the HD relay and the infinite relay antenna number were considered in \cite{Xiaojunzheng}.

\subsection{Composite Channel Estimation (CCE)}
Assume that all users transmit pilot signals simultaneously and the two users in the $n$-th user pair transmit the same pilot sequence $\sqrt{\tau_c P_p}\phi_{cn} \in \mathbb{C}^{1\times \tau_c}$ ($\phi_{cn}\phi_{cn}^H=1$, $n=1,2,\cdots ,K$), the received signal matrices of the receive and transmit antenna array of the relay are shown as
\begin{align}
&\mathbf{Y}_{\mathrm{rc}}
= \sum_{n=1}^{K}{ \sqrt{\tau_c P_p}\left( \mathbf{g}_{\mathrm{2n-1}}+\mathbf{g}_{\mathrm{2n}} \right) } \phi_{cn} + \mathbf{\bar{Z}}_{\mathrm{rc}} \nonumber\\
&\hspace{5.5mm}=\sqrt{\tau_c P_p}\mathbf{G}_\mathrm{c}\mathbf{\Phi}_\mathrm{c} + \mathbf{\bar{Z}}_{\mathrm{rc}}, \label{cce1}\\
&\mathbf{Y}_{\mathrm{tc}}
=\sqrt{\tau_c P_p}\mathbf{F}_\mathrm{c}\mathbf{\Phi}_\mathrm{c} + \mathbf{\bar{Z}}_{\mathrm{tc}}, \label{cce2}
\end{align}
respectively,
where $\mathbf{\Phi}_\mathrm{c} {=} \left[ \phi_{c1}^T, \phi_{c2}^T, \cdots, \phi_{cK}^T \right]^T {\in} \mathbb{C}^{K\times \tau_c}$ and $\mathbf{\Phi}_\mathrm{c}\mathbf{\Phi}_\mathrm{c}^H {=} \mathbf{I}_\mathrm{K}$ ($\tau_c\geqslant K$).
Let $\mathbf{G}_\mathrm{1} = [\mathbf{g}_\mathrm{1}, \mathbf{g}_\mathrm{3}, \cdots, \mathbf{g}_\mathrm{2K-1}]$, $\mathbf{G}_\mathrm{2} = [\mathbf{g}_\mathrm{2}, \mathbf{g}_\mathrm{4}, \cdots, \mathbf{g}_\mathrm{2K}]$,
 $\mathbf{F}_\mathrm{1} = [\mathbf{f}_\mathrm{1}, \mathbf{f}_\mathrm{3}, \cdots, \mathbf{f}_\mathrm{2K-1}]$ and $\mathbf{F}_\mathrm{2} = [\mathbf{f}_\mathrm{2}, \mathbf{f}_\mathrm{4}, \cdots, \mathbf{f}_\mathrm{2K}]$, thus $\mathbf{G}_\mathrm{c} = \mathbf{G}_\mathrm{1} + \mathbf{G}_\mathrm{2}$ and $\mathbf{F}_\mathrm{c} = \mathbf{F}_\mathrm{1} + \mathbf{F}_\mathrm{2}$. Besides, $\mathbf{\bar{Z}}_{\mathrm{rc}}\in \mathbb{C}^{N_r\times \tau_c}$ and $\mathbf{\bar{Z}}_{\mathrm{tc}}\in \mathbb{C}^{N_t\times \tau_c}$ denote the AWGN matrices  with each element's variance of $\sigma_{nr}^2$.

Then we obtain the LS estimations of $\mathbf{G}_\mathrm{c}$ and $\mathbf{F}_\mathrm{c}$ as
\begin{align}
&\mathbf{\hat{G}}_\mathrm{c} = \frac{1}{\sqrt{\tau_c P_p}} \mathbf{Y}_{\mathrm{rc}} \mathbf{\Phi}_\mathrm{c}^H = \mathbf{G}_\mathrm{c} + \mathbf{Z}_\mathrm{rc}, \label{cc3} \\
&\mathbf{\hat{F}}_\mathrm{c} = \frac{1}{\sqrt{\tau_c P_p}} \mathbf{Y}_{\mathrm{tc}} \mathbf{\Phi}_\mathrm{c}^H = \mathbf{F}_\mathrm{c} + \mathbf{Z}_\mathrm{tc}, \label{cc4}
\end{align}
respectively, where $\mathbf{Z}_\mathrm{rc} = \frac{1}{\sqrt{\tau_c P_p}} \mathbf{\bar{Z}}_{\mathrm{rc}} \mathbf{\Phi}_\mathrm{c}^H$ and $\mathbf{Z}_\mathrm{tc} = \frac{1}{\sqrt{\tau_c P_p}} \mathbf{\bar{Z}}_{\mathrm{tc}} \mathbf{\Phi}_\mathrm{c}^H$ signify the error matrices and their elements are all $\mathcal{CN} (0, \frac{\sigma_{nr}^2}{\tau_c P_p})$ random variables.
We observe that $\mathbf{G}_\mathrm{c}$, $\mathbf{Z}_\mathrm{rc}$,
$\mathbf{F}_\mathrm{c}$ and $\mathbf{Z}_\mathrm{tc}$ are pairwise independent.
Besides, we can easily get that
\begin{align}
&\mathbb{E}\left[\mathbf{G}_\mathrm{c}^H\mathbf{G}_\mathrm{c}\right] = N_r\left(\mathbf{D}_{\mathrm{u1}} + \mathbf{D}_{\mathrm{u2}}\right), \label{cc5}\\
&\mathbb{E}\left[\mathbf{F}_\mathrm{c}^H\mathbf{F}_\mathrm{c}\right] = N_r\left(\mathbf{D}_{\mathrm{d1}} + \mathbf{D}_{\mathrm{d2}}\right), \label{cc6}
\end{align}
where
$\mathbf{D}_{\mathrm{u1}} = \mathrm{diag}\left[\beta_{u1}, \beta_{u3}, \cdots, \beta_{u(2K-1)}\right]$,
$\mathbf{D}_{\mathrm{u2}} = \mathrm{diag}\left[\beta_{u2}, \beta_{u4}, \cdots, \beta_{u(2K)}\right]$,
$\mathbf{D}_{\mathrm{d1}} = \mathrm{diag}\left[\beta_{d1}, \beta_{d3}, \cdots, \beta_{d(2K-1)}\right]$, and
$\mathbf{D}_{\mathrm{d2}} = \mathrm{diag}\left[\beta_{d2}, \beta_{d4}, \cdots, \beta_{d(2K)}\right]$.
Therefore, the covariance matrices of the rows of $\mathbf{\hat{G}}_\mathrm{c}$ and $\mathbf{\hat{F}}_\mathrm{c}$ are denoted as
\begin{align}
&\mathbf{\hat{D}}_{\mathrm{uc}} = \frac{\mathbb{E}\left[\mathbf{\hat{G}}_\mathrm{c}^H\mathbf{\hat{G}}_\mathrm{c}\right]}
{N_r}
=\mathbf{D}_{\mathrm{u1}}+\mathbf{D}_{\mathrm{u2}}+\frac{\sigma_{nr}^2}{\tau_c P_p} \mathbf{I}_{\mathrm{K}}, \label{cc7}\\
&\mathbf{\hat{D}}_{\mathrm{dc}} = \frac{\mathbb{E}\left[\mathbf{\hat{F}}_\mathrm{c}^H\mathbf{\hat{F}}_\mathrm{c}\right]}
{N_t}
=\mathbf{D}_{\mathrm{d1}}+\mathbf{D}_{\mathrm{d2}}+\frac{\sigma_{nr}^2}{\tau_c P_p} \mathbf{I}_{\mathrm{K}}, \label{cc8}
\end{align}
where the $n$-th diagonal elements of $\mathbf{\hat{D}}_{\mathrm{uc}}$ and $\mathbf{\hat{D}}_{\mathrm{dc}}$ are represented as $\hat{\beta}_{ucn}$ and $\hat{\beta}_{dcn}$, respectively.

With respect to the training length, \cite{howMuchTraining} shows that the optimal training length equals the minimum possible, i.e., $\tau=2K$ and $\tau_c=K$, assuming that the training power and data power can vary. However, when the training power and the data power are equal and very low, the optimal number of training symbols may be larger.
Without loss of generality, we use $\tau_c=\frac{1}{2}\tau$ in the following.

\begin{corollary}\label{cor1}
When $\tau_c = \frac{1}{2}\tau$, based on (\ref{iCE5}), (\ref{iCE6}) and (\ref{cc7}), (\ref{cc8}), we can easily get that
\begin{align}
&\hat{\beta}_{ucn} = \hat{\beta}_{u(2n-1)} + \hat{\beta}_{u(2n)}, \label{cc9} \\
&\hat{\beta}_{dcn} = \hat{\beta}_{d(2n-1)} + \hat{\beta}_{d(2n)}, \label{cc10}
\end{align}
for any user pair $n$ ($n=1,2,\cdots,K$).
\end{corollary}

\subsection{Achievable Rate with CCE}
With the estimated composite channels, the relay takes the following MRC/MRT matrix
\begin{equation}\label{arwcce1}
\mathbf{W}_\mathrm{c} = \mathbf{\hat{F}}_\mathrm{c}^*\mathbf{\hat{G}}_\mathrm{c}^H
=\left(\mathbf{F}_\mathrm{c}+\mathbf{Z}_\mathrm{tc}\right)^*
\left(\mathbf{G}_\mathrm{c}+\mathbf{Z}_\mathrm{rc}\right)^H.
\end{equation}

Similar to (\ref{eq16}) and (\ref{eq17}), substituting (\ref{arwcce1}) into (\ref{eq5}) and (\ref{eq6}), we get
\begin{align}
&\Delta_1 {=} N_t\sum_{n=1}^{K}{ \hat{\beta}_{dn}\left[N_r^2(\beta_{u(2n{-}1)}^2{+}\beta_{u(2n)}^2) {+} N_r\hat{\beta}_{un}\sum_{j=1}^{2K}{ \beta_{uj} } \right] }, \label{arwcce2}\\
&\Delta_2 = N_tN_r\sum_{n=1}^{K}{ \hat{\beta}_{dn}\hat{\beta}_{un} }. \label{arwcce3}
\end{align}
Then the power limiting factor $\alpha$ with CCE is achieved by substituting (\ref{arwcce2}) and (\ref{arwcce3}) into (\ref{eq4}).

\begin{theorem}\label{theorem2}
Without loss of generality, consider user $S_k$ ($k=2m-1$) in user pair $m$. When $\kappa$ is fixed and $N_r\gg 2K$, the SINR of user $S_k$ for a finite number of relay antennas under CCE is approximated as
\begin{equation}\label{arwcce4}
\gamma_k^c \approx  \frac{N_t}{A_k^c + {\text{MP}}_k^c+{\text{LIR}}_k^c+{\text{NR}}_k^c+{\text{MU}}_k^c+{\text{AN}}_k^c},
\end{equation}
where
\begin{align}
&A_k^c = \kappa \frac{\hat{\beta}_{ucm}}{\beta_{u(2m)}} + \frac{\hat{\beta}_{dcm}}{\beta_{d(2m-1)}}, \label{arwcce5}\\
&\text{MP}_k^c =
\sum_{n=1, n\neq m}^{K}  \kappa\frac{\hat{\beta}_{ucm}(\beta_{u(2n-1)}+\beta_{u(2n)})}{\beta_{u(2m)}^2}
\nonumber\\
&\hspace{20mm}+\sum_{n=1, n\neq m}^{K}
\frac{\hat{\beta}_{dcn}(\beta_{u(2n-1)}^2+\beta_{u(2n)}^2)}{\beta_{d(2m-1)}\beta_{u(2m)}^2}, \label{arwcce6} \\
&\text{LIR}_k^c = \frac{P_R\sigma_{LI}^2}{P_S}\frac{\kappa\hat{\beta}_{ucm}}{\beta_{u(2m)}^2}, \label{arwcce7} \\
&\text{NR}_k^c = \frac{\sigma_{nr}^2}{P_S}\frac{\kappa\hat{\beta}_{ucm}}{\beta_{u(2m)}^2}, \label{arwcce8} \\
&{\text{MU}}_k^c = \frac{1}{\beta_{d(2m-1)}^2\beta_{u(2m)}^2}
\Delta_3^c \sum_{i\in \text{U}_k}{\sigma_{k,i}^2} , \label{arwcce9}
\end{align}
\begin{align}
&{\text{AN}}_k^c = \frac{\sigma_n^2}{P_S\beta_{d(2m-1)}^2\beta_{u(2m)}^2} \Delta_3^c ,\label{arwcce10}
\end{align}
%\begin{align}
%&A_k^c = \kappa \frac{\hat{\beta}_{ucm}}{\beta_{u(2m)}} + \frac{\hat{\beta}_{dcm}}{\beta_{d(2m-1)}}, \label{arwcce5}\\
%&\text{MP}_k^c =
%\sum_{n=1, n\neq m}^{K}  \kappa\frac{\hat{\beta}_{ucm}(\beta_{u(2n-1)}+\beta_{u(2n)})}{\beta_{u(2m)}^2}
%\nonumber\\
%&\hspace{20mm}+\sum_{n=1, n\neq m}^{K}
%\frac{\hat{\beta}_{dcn}(\beta_{u(2n-1)}^2+\beta_{u(2n)}^2)}{\beta_{d(2m-1)}\beta_{u(2m)}^2}, \label{arwcce6} \\
%&\text{LIR}_k^c = \frac{P_R\sigma_{LI}^2}{P_S}\frac{\kappa\hat{\beta}_{ucm}}{\beta_{u(2m)}^2}, \label{arwcce7} \\
%&\text{NR}_k^c = \frac{\sigma_{nr}^2}{P_S}\frac{\kappa\hat{\beta}_{ucm}}{\beta_{u(2m)}^2}, \label{arwcce8}\\
%&{\text{MU}}_k^c = \frac{1}{\beta_{d(2m-1)}^2\beta_{u(2m)}^2}
%\Delta_3^c \sum_{i,k\in \text{U}_k}{\sigma_{k,i}^2} , \label{arwcce9}\\
%&{\text{AN}}_k^c = \frac{\sigma_n^2}{P_S\beta_{d(2m-1)}^2\beta_{u(2m)}^2} \Delta_3^c ,\label{arwcce10}
%\end{align}
and\footnote{Similar to the fourth footnote, considering the power-scaling law, we have $\Delta_3^c{=}\frac{P_S}{P_R}\sum\limits_{n=1}^{K}{\hat{\beta}_{dcn}(\beta_{u(2n{-}1)}^2{+}\beta_{u(2n)}^2)} {+} \frac{\sigma_{nr}^2}{N_rP_R} \sum\limits_{n=1}^{K}{\hat{\beta}_{dcn}\hat{\beta}_{ucn}}$.}
$\Delta_3^c = \frac{P_S}{P_R}\sum\limits_{n=1}^{K}{\hat{\beta}_{dcn}(\beta_{u(2n{-}1)}^2 + \beta_{u(2n)}^2)}$.
\end{theorem}
\begin{proof}
The proof is similar with Theorem 1.
\end{proof}

We observe that the SINR of user $S_k$ with CCE is similar to that with ICE. By comparing (\ref{arwcce4}) with (\ref{eq19}) and based on (\ref{cc9}) and (\ref{cc10}), we obtain $\gamma_k^c < \gamma_k$. Particularly, consider that the system is symmetric\footnote{It is known that the large-scale fading is closely related to the distance, hence symmetry here can be interpreted as the same distance from the two users in each user pair to the relay.}, i.e. $\mathbf{D}_{\mathrm{u1}}=\mathbf{D}_{\mathrm{u2}}$ and  $\mathbf{D}_{\mathrm{d1}}=\mathbf{D}_{\mathrm{d2}}$, then we can easily obtain $\gamma_k^c =\frac{1}{2}\gamma_k$, $\forall k$.

\begin{remark}\label{remarkCCESIC}
As to the self-interference under CCE, according to the law of large numbers (\textit{Lemma 1} in \cite{massivemimo2}) and in the regime of very large $N_r$ and $N_t$, we get
\begin{align}\label{SIcce1}
\mathbf{f}_\mathrm{k}^T\mathbf{W}_\mathrm{c}\mathbf{g}_\mathrm{k}
& {=} \mathbf{f}_\mathrm{k}^T  \sum_{n=1}^{K}{\left(\mathbf{f}_\mathrm{2n-1}^*{+}\mathbf{f}_\mathrm{2n}^*{+}\mathbf{z}_{\mathrm{tcn}}^*\right) \left(\mathbf{g}_\mathrm{2n-1}^H{+}\mathbf{g}_\mathrm{2n}^H{+}\mathbf{z}_{\mathrm{rcn}}^H\right)}   \mathbf{g}_\mathrm{k} \nonumber\\
& \approx \left\|\mathbf{f}_\mathrm{k}\right\|_2^2\left\|\mathbf{g}_\mathrm{k}\right\|_2^2
\approx N_tN_r\beta_{dk}\beta_{uk},
\end{align}
where $\mathbf{z}_{\mathrm{tcn}}$ and $\mathbf{z}_{\mathrm{rcn}}$ are the $n$-th columns of $\mathbf{Z}_{\mathrm{tc}}$ and $\mathbf{Z}_{\mathrm{rc}}$, respectively.
Note that the individual CSI for each user cannot be acquired under CCE. But based on the law of large numbers, we can approximate the self-interference $\alpha\sqrt{P_S}\mathbf{f}_\mathrm{k}^T\mathbf{W}_\mathrm{c}\mathbf{g}_\mathrm{k}x_k$ as $\alpha\sqrt{P_S}N_tN_r\beta_{dk}\beta_{uk}x_k$. It is meant that the self-interference is only related to the large-scale fading coefficients and then can be cancelled out.
\end{remark}

\section{Performance Evaluation}\label{performanceEvaluation}
In this section, we evaluate the system performance with different pilot schemes. First, we analytically compare the performance under ICE with that under CCE. Then, the power control of the users and the relay is derived based on sum rate maximization and max-min fairness criterion, respectively.

\subsection{Performance Comparison Between ICE and CCE}
The previous analysis shows that in the symmetric system ($\mathbf{D}_{\mathrm{u1}}=\mathbf{D}_{\mathrm{u2}}$,  $\mathbf{D}_{\mathrm{d1}}=\mathbf{D}_{\mathrm{d2}}$),
we have $\gamma_k^c=\frac{1}{2}\gamma_k$, $\forall k$. Then we can obtain the following corollary.
\begin{corollary}\label{cor2}
Consider the symmetric traffic and $\tau_c = \frac{1}{2}\tau$.
Let $T_c^E=\left( 1+\frac{1}{2(g-1)} \right)\tau$, where $g=\frac{ \sum_{k=1}^{2K} \log_2(1+\gamma_k)}{ \sum_{k=1}^{2K} \log_2(1+\frac{1}{2}\gamma_k)}$. The coherence interval $T_c^E$ satisfies $R^c=R$, where $R^c$ denotes the sum rate of the system with CCE. Moreover, we have
\begin{itemize}
  \item $R^c>R$, when $T_c\in( \tau, T_c^E )$;
  \item $R^c<R$, when $T_c\in( T_c^E, \infty )$.
\end{itemize}
\end{corollary}
\begin{proof}
In the symmetric system, when $\tau_c=\frac{1}{2}\tau$, we have
\begin{equation}\label{proofCor2}
\frac{R^c}{R}
=\frac{T_c-\frac{1}{2}\tau}{T_c-\tau} \cdot \left(\frac{ \sum_{k=1}^{2K} \log_2(1+\gamma_k)}{ \sum_{k=1}^{2K} \log_2(1+\frac{1}{2}\gamma_k)}\right)^{-1}.
\end{equation}
Let $g=\frac{ \sum_{k=1}^{2K} \log_2(1+\gamma_k)}{ \sum_{k=1}^{2K} \log_2(1+\frac{1}{2}\gamma_k)}$ ($\gamma_k>0$, $\forall k$). Evidently, we have $g>1$.  In addition, it can be easily proved that $\frac{1}{2}\log_2(1+\gamma_k)<\log_2(1+\frac{1}{2}\gamma_k)$, $\forall k$, therefore, we obtain $g\in(1,2)$.

Besides, let $f(T_c) = \frac{R^c}{R} = \frac{T_c-\frac{1}{2}\tau}{g\left(T_c-\tau\right)}$. We can easily show that $f(T_c)$ is a monotonically decreasing function of $T_c$. Let $f(T_c) = 1$, then we get the solution of $T_c$:
$T_c^E = \left( 1+\frac{1}{2(g-1)} \right)\tau$. Since $g\in(1,2)$, we have $T_c^E>\frac{3}{2}\tau >\tau$.
Based on the monotonically decreasing property of $f(T_c)$, we have $f(T_c)>1$ if $\tau<T_c<T_c^E$, and $f(T_c)<1$ if $T_c>T_c^E$. Then the proof is completed.
\end{proof}

Corollary \ref{cor2} shows that the CCE scheme performs better in the scenario where the coherence interval is smaller than a certain value. Otherwise, the ICE scheme is preferable.

In addition, we rewrite (\ref{eq19}) as $\gamma_k = \theta_k N_t$, where $\theta_k \approx \frac{1}{A_k + {\text{MP}}_k+{\text{LIR}}_k+{\text{NR}}_k+{\text{MU}}_k+{\text{AN}}_k}$. Then we have the following corollary concerning the relation between $T_c^E$ and $N_t$.
\begin{corollary}\label{relationTcENt1}
When the number of relay antennas is very large such that $\theta_kN_t \gg 2$ for any user, $T_c^E$ is increasing with respect to $N_t$ in an approximately logarithmic way.
\end{corollary}
\begin{proof}
Consider $N_t$ as the argument and keep other system parameters fixed, such as the transmit powers, large-scale fading and the interference levels, thus all $\theta_k$ ($\forall k$) are constant and positive values. Then we have
\begin{equation}\label{gResponse1}
g(N_t)=\frac{ \sum_{k=1}^{2K} \log_2(1+\theta_k N_t)}{ \sum_{k=1}^{2K} \log_2(1+\frac{1}{2}\theta_k N_t)}.
\end{equation}
When $N_t$ is very large such that $\theta_kN_t \gg 2$, $\forall k$, we get
\begin{align}\label{approxgNt1}
g(N_t) &\approx \frac{ \sum_{k=1}^{2K} \log_2(\theta_k N_t)}{ \sum_{k=1}^{2K} \log_2(\frac{1}{2}\theta_k N_t)} \nonumber\\
& = 1 + \frac{1}{\log_2{N_t} + \left(\sum_{k=1}^{2K} \log_2\theta_k \right) / (2K) - 1},
\end{align}
which shows that $g$ is logarithmically decreasing with $N_t$ when $N_t$ is very large.
In addition, $T_c^E$ is decreasing with $g$. As a result, we obtain that $T_c^E$ is increasing with respect to $N_t$ in an approximately logarithmic way in massive MIMO.
\end{proof}

\begin{remark}\label{remarkComplexityanalysis}
(\textit{Complexity Analysis}) Since both the two channel estimation schemes adopt the LS estimation approach, the difference of complexity between the two schemes only stays in the channel estimation stage (i.e., in (\ref{iCE3}), (\ref{iCE4}) and (\ref{cc3}), (\ref{cc4})) and the computing stage of the MRC/MRT matrix (i.e., in (\ref{eq15}) and (\ref{arwcce1})), thus we only present the complexity analysis of these two stages. Thereby, the time complexity of ICE scheme is given by $\mathcal{O}(2KN_tN_r) + \mathcal{O}(2K\tau (N_t+N_r))$, and the time complexity of CCE scheme is shown as $\mathcal{O}(KN_tN_r) + \mathcal{O}(K\tau_c (N_t+N_r))$. In particular, when $\tau_c = \frac{1}{2}\tau = K$, the ICE scheme has an extra time complexity of $\mathcal{O}(KN_tN_r) + \mathcal{O}(3K^2 (N_t+N_r))$.
\end{remark}

\subsection{Power Control}
Without loss of generality, we only present the power allocation of the system under ICE. And we consider the fixed pilot power $P_p$.
First, different transmit powers are optimally allocated to different users and the relay in order to obtain the maximal achievable sum rate. Then we address the power allocation problem for maximizing the minimum SINR of all the users. In the end, based on the max-min fairness criterion, we discuss a special scenario where all large-scale fading coefficients are set to be the same.

\subsubsection{Sum Rate Maximization} Assuming different users have different transmit powers and user $S_i$ takes the transmit power $P_i$. Then we can obtain the approximate SINR of $S_k$ using the similar way as that of Theorem \ref{theorem1}.
\begin{theorem}\label{theorem3}
When $\kappa$ is fixed and finite $N_r$ satisfies $N_r\gg 2K$, the approximate SINR of user $S_k$ under different user powers is given by
\begin{align}\label{plofdiffuserP1}
\gamma_k
&\approx \frac{1}{f_k(P_1,P_2,\cdots,P_{2K},P_R)} \nonumber\\
&=\frac{P_{k'} N_t}{ {\overline{\text{MP}}}_k+{\overline{\text{LIR}}}_k+{\overline{\text{NR}}}_k+{\overline{\text{MU}}}_k+{\overline{\text{AN}}}_k},
\end{align}
where
\begin{align}
&\overline{\text{MP}}_k =
\sum_{j=1, j\neq k}^{2K} {P_j a_{k,j} }, \label{plofdiffuserP2} \\
&\overline{\text{LIR}}_k = P_R \sigma_{LI}^2 b_k, ~~~ \overline{\text{NR}}_k = \sigma_{nr}^2 b_k, \label{plofdiffuserP3} \\
% &\text{NR}_k = \sigma_{nr}^2 b_k, \label{plofdiffuserP4}\\
&{\overline{\text{MU}}}_k =
 \frac{1}{P_R}\sum_{i=1}^{2K}{P_{i} c_{k,i}}
\sum_{ i\in\mathrm{U}_{k} }{P_i \sigma_{k,i}^2} , \label{plofdiffuserP5}\\
&{\overline{\text{AN}}}_k =
\frac{\sigma_n^2}{P_R}\sum_{i=1}^{2K}{P_{i} c_{k,i}} . \label{plofdiffuserP6}
\end{align}
and $a_{k,j}=\kappa\frac{\beta_{uj}\hat{\beta}_{uk'}}{\beta_{uk'}^2}
+\frac{\hat{\beta}_{dj'}\beta_{uj}^2}{\beta_{dk}\beta_{uk'}^2}$,
$b_k =  \frac{\kappa\hat{\beta}_{uk'}}{\beta_{uk'}^2}$,
$c_{k,i} = \frac{\hat{\beta}_{di'}\beta_{ui}^2}{\beta_{dk}^2\beta_{uk'}^2}$.
\end{theorem}

Our goal of power allocation is to maximize the sum rate, and this optimization problem is formulated as
\begin{align}
&\max_{\{P_i,P_R,\gamma_k\}} ~ \prod_{k=1}^{2K} (1+\gamma_k) \nonumber\\
&\hspace{1mm}\mathrm{s.t.} ~~ \gamma_k \cdot f_k(P_1,P_2,\cdots,P_{2K},P_R) \leqslant 1, \forall k, \nonumber\\
&\hspace{9mm}0\leqslant P_i \leqslant P_S^{max}, \forall i, \nonumber\\
&\hspace{9mm}0\leqslant P_R \leqslant P_R^{max},  \label{optbro1}
\end{align}
where $P_S^{max}$ and $P_R^{max}$ are  peak power constraints of the users and the relay, respectively.
Note that we use the inequality constraints in (\ref{optbro1}) to replace the equality constraints in (\ref{plofdiffuserP1}) \cite{approximationOnePlusR}, and the optimal solution of (\ref{optbro1}) will satisfy (\ref{plofdiffuserP1}) since the objective function increases with each $\gamma_k$.

Equation (\ref{plofdiffuserP1}) indicates that the function $f_k(P_1,P_2,\cdots,P_{2K},P_R)$ is a posynomial, then the problem (\ref{optbro1}) would be a geometric program (GP) if the objective function was a monomial. To solve a GP which can be converted to convex form, we can use CVX (a convex optimization tool) \cite{cvx}. Based on the Lemma 1 in \cite{approximationOnePlusR}, we can approximate $1+\gamma_k$ as
\begin{equation}\label{approOnePlusR}
1+\gamma_k \approx \lambda_k \gamma_k^{\nu_k}
\end{equation}
near an arbitrary point $\hat{\gamma}_k>0$, where $\nu_k = \hat{\gamma}_k(1+\hat{\gamma}_k)^{-1}$ and $\lambda_k = \hat{\gamma}_k^{-\nu_k}(1+\hat{\gamma}_k)$. Note that $1+\gamma_k = \lambda_k \gamma_k^{\nu_k}$ if and only if $\gamma_k=\hat{\gamma}_k$, and $1+\gamma_k > \lambda_k \gamma_k^{\nu_k}$ otherwise \cite{approximationOnePlusR}.
Hence, by making a guess of $\hat{\gamma}_k$, (\ref{optbro1}) can be approximated as
\begin{align}
&\max_{\{P_i,P_R,\gamma_k\}} ~ \prod_{k=1}^{2K} \lambda_k \gamma_k^{\nu_k} \nonumber\\
&\hspace{1mm}\mathrm{s.t.} ~~ \gamma_k \cdot f_k(P_1,P_2,\cdots,P_{2K},P_R) \leqslant 1, \forall k, \nonumber\\
&\hspace{9mm} 0\leqslant P_i \leqslant P_S^{max}, \forall i, \nonumber\\
&\hspace{9mm} 0\leqslant P_R \leqslant P_R^{max}.   \label{optbro2}
\end{align}
%where the parameter $\rho>1$ controls the approximation accuracy of (\ref{approOnePlusR}) and little $\rho$ closing to 1 results in high accuracy\footnote{The parameter $\rho$ also affects the convergence speed of the following Algorithm \ref{algorithm1}, little $\rho$ will induce slow convergence speed. Generally, a good tradeoff between the accuracy and the convergence speed is achieved when $\rho = 1.1$ \cite{approximationOnePlusR}.}.

We can see that problem (\ref{optbro2}) is a GP and can be solved by CVX.
Note that  (\ref{optbro2}) should be solved several times to refine the solution by updating the last solution of $\gamma_k$ as $\hat{\gamma}_k$.  The iteration algorithm for solving (\ref{optbro1}) is formulated as follows\footnote{Note that the approximation (\ref{approOnePlusR}) requires $\hat{\gamma}_k>0$, thus we will stop the iteration if zero arises in the solution of $\gamma_k$.}.
\begin{algorithm}[htb]
\caption{: Successive approximation algorithm for (\ref{optbro1}).}
\label{algorithm1}
\begin{algorithmic}[1]                % 1: display number before each line
\STATE \textit{Initialization}. Set the iteration number $i=1$. Choose a feasible power allocation, such as $P_k=P_S^{max}$, $\forall k$ and $P_R = KP_S^{max}$. The initial guess of $\hat{\gamma}_{k,i}$ for all $k$ is determined by (\ref{plofdiffuserP1}). Besides, set a convergence judgement parameter $\epsilon>0$.  \label{initialization}
\STATE \textit{Iteration i}. Solve the GP (\ref{optbro2}) using CVX. The solution of $\gamma_k$ is denoted as $\gamma_k^{\ast}$, $\forall k$. \label{IterationI}
\STATE \textit{Zero judgement}. If there exists any $k$ satisfying $\gamma_k^*=0$, stop. Otherwise, go to step \ref{convergenceJudge}. \label{zeroJudge}
\STATE \textit{Convergence judgement}. If $\max_{k} | \hat{\gamma}_{k,i} - \gamma_k^* | \leqslant \epsilon$, stop. Otherwise, go to step \ref{returnStep}. \label{convergenceJudge}
\STATE Set $i=i+1$, $\hat{\gamma}_{k,i} = \gamma_k^*$, $\forall k$, then go to step 2. \label{returnStep}
\end{algorithmic}
\end{algorithm}

Similar to the Theorem 1 in \cite{approximationOnePlusR}, it is easy to obtain
\begin{align}\label{monoIncreasing}
\prod_{k=1}^{2K} (1+\hat{\gamma}_{k,i})
&= \prod_{k=1}^{2K} \lambda_{k,i}\hat{\gamma}_{k,i}^{\nu_{k,i}}
\leqslant \prod_{k=1}^{2K} \lambda_{k,i}\hat{\gamma}_{k,(i+1)}^{\nu_{k,i}} \nonumber\\
&\leqslant \prod_{k=1}^{2K} (1+\hat{\gamma}_{k,(i+1)}),
\end{align}
where the first inequality follows from the fact that $\hat{\gamma}_{k,(i+1)}$ is the solution of (\ref{optbro2}) after the $i$-th iteration, and the second inequality follows from $1+\gamma_k\geqslant \lambda_k \gamma_k^{\nu_k}$. Thus Algorithm \ref{algorithm1} is monotonically increasing with the iteration number\footnote{In practice, we always stop the algorithm after a few iterations before it converges. The increasing property guarantees the effectiveness of the algorithm.}.

\subsubsection{Max-Min Fairness Criterion}
On the other hand, since each user receives data from its partner, the rate performance of the worst user with the minimum SINR represents the system performance in some sense. As a result, based on the max-min fairness design criterion, the goal of power allocation can be maximization of the minimum SINR of the users. This optimization problem is formulated as follows.
\begin{align}
\max_{\{P_i, P_R\}} &~ \min_{\forall k} ~~  \frac{1}{f_k(P_1,P_2,\cdots,P_{2K},P_R)} \nonumber\\
&\hspace{1mm}\mathrm{s.t.} ~~  0\leqslant P_i \leqslant P_S^{max}, \forall i, \nonumber\\
&\hspace{8.2mm} 0\leqslant P_R \leqslant P_R^{max}.  \label{optmaxminprob1}
\end{align}

Similar to the technique in \cite{zhengzhengxiang}, a slack variable $t$ is introduced to (\ref{optmaxminprob1}) and the problem is reformulated as
\begin{align}
&\max_{\{P_i, P_R, t \}} ~ t   \nonumber\\
&\hspace{1mm}\mathrm{s.t.} ~~ t\cdot f_k(P_1,P_2,\cdots,P_{2K},P_R) \leqslant 1, \forall k, \nonumber\\
&\hspace{8.2mm} 0\leqslant P_i \leqslant P_S^{max}, \forall i, \nonumber\\
&\hspace{8.2mm} 0\leqslant P_R \leqslant P_R^{max}.  \label{optmaxminprob3}
\end{align}

We observe that the optimization problem (\ref{optmaxminprob3}) is a GP and can be solved using CVX.

\subsubsection{A Special Scenario}
In this special scenario, we assume that all large-scale fading coefficients are the same (i.e., $\mathbf{D}_\mathrm{u}=\mathbf{D}_\mathrm{d}=\beta \mathbf{I}_{\mathrm{2K}}$, there would be no large-scale fading if $\beta=1$).
Without loss of generality, we consider equal self-LI levels and equal inter-user interference levels, i.e., $\sigma_{k,k}^2=\sigma_{LI}^2$, $\sigma_{k,i}^2=\sigma_{IU}^2$, $\forall k$ and $i\neq k$.
From (\ref{plofdiffuserP2}) $\sim$ (\ref{plofdiffuserP6}), we obtain $a_{k,j}=(\kappa+1)\frac{\hat{\beta}}{\beta}$, $b_k= \frac{\kappa\hat{\beta}}{\beta^2}$, $c_{k,i}=\frac{\hat{\beta}}{\beta^2}$, $\forall k,j,i$. Then if all users transmit the same power ($P_i=P_S$, $\forall i$), the SINRs for all users will be equal,  which satisfies the condition of the max-min fairness criterion.

With all users having the same power, (\ref{eq21}) $\sim$ (\ref{eq24}) are re-expressed as
\begin{align}
&{\text{LIR}}_k = \frac{P_R}{P_S} \cdot \sigma_{LI}^2 b, ~~~ \text{NR}_k = \frac{1}{P_S}\sigma_{nr}^2 b, \label{eq25}\\
&{\text{MU}}_k =  \frac{P_S}{P_R} \cdot 2K^2 \mu \delta, \label{eq26}\\
&{\text{AN}}_k =  \frac{1}{P_R} \cdot 2K\sigma_n^2 \mu, \label{eq27}
\end{align}
where $b {=} \frac{\kappa\hat{\beta}}{\beta^2}$,
$\mu {=} \frac{\hat{\beta}}{\beta^2}$ and
$\delta {=} \frac{\sigma_{LI}^2 + (K-1)\sigma_{IU}^2}{K}$. Accordingly,
the max-min problem is equivalent to the following minimization problem:
\begin{align}
&\min_{\{ P_S, P_R \}} ~ f_k(P_S, P_R) = \text{LIR}_k + \text{NR}_k + \text{MU}_k + \text{AN}_k \nonumber\\
&\hspace{4mm}\mathrm{s.t.} ~~ 0 \leqslant P_S\leqslant P_S^{max}, P_R \geqslant 0,  \label{eq28}
\end{align}
where $P_S^{max}$ is the peak power constraint of $P_S$.
By solving (\ref{eq28}), we obtain the following corollary.
\begin{corollary}\label{corPowerControl1}
Let $\eta {=} K \sqrt{ \frac{2 \delta}{\sigma_{LI}^2 \kappa} } $. When $K\delta P_S^{max} {\gg} \sigma_n^2$, the optimal value of the max-min problem for the special case is achieved when $P_S=P_S^{max}$ and $P_R \approx \eta P_S^{max}$.
\end{corollary}
\begin{proof}
Equation (\ref{eq28}) is the minimization problem about a binary function. Thus by utilizing the properties of binary functions, we can get that the optimal value of (\ref{eq28}) is achieved when
$P_S=P_S^{max}$ and
$P_R=\sqrt{\frac{2K\mu P_S^{max}}{\sigma_{LI}^2 b} \left(K\delta P_S^{max}+\sigma_n^2\right) }$. Then we can obtain Corollary \ref{corPowerControl1}.
\end{proof}

Corollary \ref{corPowerControl1} indicates the simple power allocation based on max-min fairness criterion when all large-scale fading coefficients are the same. In addition, we observe that the variable $\eta$ is decided by interference levels ($\sigma_{LI}^2$ and $\sigma_{IU}^2$), while it has nothing to do with the large-scale fading $\beta$. Furthermore, the numerical results by operating Algorithm 1 can show that in this special case, the sum rate optimization criterion yields the same power allocation results as that of max-min criterion.

\begin{remark}\label{interferenceAnalysis}
(\textit{Interference Analysis}) From (7), we know that there are mainly four types of interference in the system, in which inter-pair interference and inter-user interference are due to the multi-pair consideration, and LI from the relay and self-LI are caused by the full-duplex operation. (22) shows that all these interferences are harmful to the SINR. With respect to their impacts on the sum rate, we know from (24) and (27) that inter-pair and inter-user interferences increase with the number of user pairs $K$. However, the previous analysis indicates that the sum rate will achieve its maximum with an optimal $K$ for a fixed $N_t$. Thus the sum rate will first increase with $K$ though the inter-pair and inter-user interferences also increase since the multiplexing gain is larger than the interference effect, and then the sum rate will decrease with $K$ when $K$ is larger than its optimal value due to the opposite reason.
In addition, the power allocation result of Corollary 4 implies that there is a nearly linear relation between $P_R$ and $P_S$ for the special scenario when the sum rate achieve its peak value.
Therefore, for LI from the relay and under fixed $P_S$, the sum rate will first increase and then decrease with $P_R$ while the LI from the relay keeps increasing with $P_R$, and the case of self-LI under fixed $P_R$ is similar. Moreover, if both $P_R$ and $P_S$ increase according to their optimal relation in Corollary 4, it can be easily inferred from (22) $\sim$ (28) that the sum rate will increase though both the interferences also increase.

From (22), we see that the SINR will increase significantly with $N_t$ due to the large array gain of massive MIMO. However, the denominator in (22) indicates that the interferences keep constant regardless of the increase of antenna number. Therefore, by enhancing the received power of the desired signal and the relative value SINR, it seems that massive MIMO is able to suppress all types of interference.
\end{remark}

\section{Numerical Results}\label{numericalResults}
In this section, the numerical results for the proposed system are illustrated, as well as the comparison of our scheme with other schemes.
We set $N_t = N_r$, the coherence interval $T_c=100$ symbols\footnote{Note that $T_c$ is inversely proportional to the Doppler spread, thus its value is large for cases of low mobility (for example, the relay is fixed geographically) and small for cases of high mobility (for example, a mobile terminal serves as the relay), and we set $T_c=100$ for cases of low mobility. Besides, both cases of high and low mobilities are considered in Fig. 4, where we set $2K<T_c\leqslant 100$.}
and the number of user pairs $K=5$ unless otherwise specified. The noise is normalized to be $\sigma^2_n=\sigma_{nr}^2=1$. In addition, the inter-user interference level is also set to be 1, i.e. $\sigma_{k,i}^2=1$ ($i\in \mathrm{U}_k, i\neq k$), and we choose the LI level $\sigma_{LI}^2/\sigma_n^2 = \sigma_{k,k}^2/\sigma_n^2 = 5$ dB ($\forall k$)  unless otherwise specified\footnote{It is reasonable that the LI is stronger than the inter-user interference, since the inter-user interference will undergo path loss and shadow fading while LI doesn't. Besides, the user power is generally not very high, therefore, the inter-user interference becomes very small after the fading. In addition, we also consider different LI levels in Fig. 3, where the LI is set to be 1 dB, 5 dB or 10 dB.}.
Besides, the training length is set to be $\tau = 2K$ for ICE and $\tau_c = K$ for CCE and we set $P_p/\sigma_n^2 = P_S/\sigma_n^2 = 10$ dB\footnote{We set the training lengths of ICE and CCE schemes to be their minimums, to render the data transmissions more symbols of the coherence interval. Besides, when power optimization is not applied, without loss of generality, the training power and the data power are set to be the same. In addition, \cite{howMuchTraining} indicates that when the training power and the data power are equal, the smallest training length may not be optimal if the power is very low. On the other hand, the power can't be very large otherwise very strong LI would be created.}.

\subsection{Sum Rate with Statistical CSI}
In this subsection, we only consider the case of ICE, and the large-scale fading coefficients are normalized to be 1, i.e. $\mathbf{D}_\mathrm{u} = \mathbf{D}_\mathrm{d} = \mathbf{I}_{\mathrm{2K}}$.

\begin{figure}
\centering
\includegraphics[width=2.5in,height=1.9in]{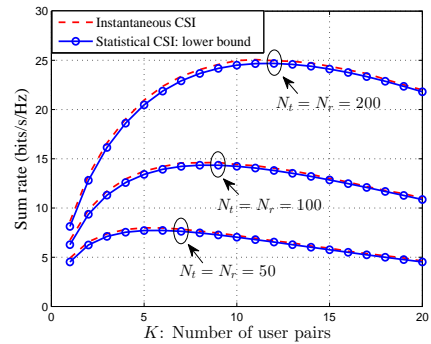}
\caption{Sum rate vs. $K$ under different $N_t$ ($P_R = \eta P_S$).}\label{SRvsKcompairInstansAndStatis}
\end{figure}

\begin{figure}
\centering
\includegraphics[width=2.5in,height=1.9in]{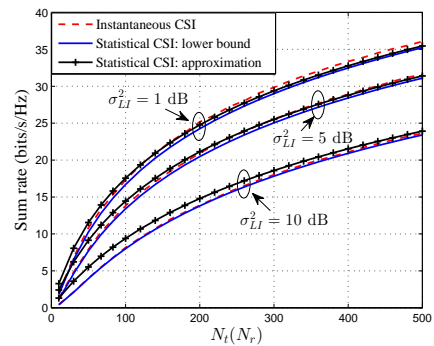}
\caption{Sum rate vs. $N_t$ under different LI level ($P_R = \eta P_S$).}\label{SRvsNtwithdiffLI}
\end{figure}

Fig. \ref{SRvsKcompairInstansAndStatis} describes the sum rate of the system vs. $K$ ($1\leqslant K\leqslant 20$) under different number of relay antennas. We set $P_R = \eta P_S$ based on Corollary \ref{corPowerControl1}. The ``Statistical CSI: lower bound'' curve is generated from (\ref{eq13}), where the statistical distributions of the channels are used to detect the desired signals. The ergodic sum rate given by (\ref{eq10}) is also presented in Fig. \ref{SRvsKcompairInstansAndStatis} by using Monte Carlo simulation, via the ``Instantaneous CSI'' curve. We can observe that the proposed rate expression (\ref{eq13}) is a lower bound of the ergodic rate, and the performance gap between them is very small, even when the relay antenna number is not so large, such as $N_t=N_r=50$, verifying that using statistical channels for signal detection in massive MIMO systems is quite feasible.
Besides, we see that the sum rate of a $K$-pair ($K>1$) system is less than $K$ times the sum rate of a one-pair system, because the multi-pair system has more interferences (inter-pair and inter-user interferences) than the one-pair system. We also observe that the sum rate increases with the increase of $K$, since the multiplexing gain is larger than the interference effect. However, when $K$ is very large, the inter-pair and inter-user interferences dominate the system performance, and hence the sum rate will decrease in very large $K$ (such as $K>5$ when $N_t=50$). Furthermore, Fig. \ref{SRvsKcompairInstansAndStatis}  verifies that adding the number of relay antennas can significantly improve the system performance.

In Fig. \ref{SRvsNtwithdiffLI}, we compare the approximate sum rate given by Theorem \ref{theorem1} with the lower bound given by (\ref{eq13}) with $P_R=\eta P_S$ and $\sigma_{LI}^2=\sigma_{k,k}^2$, $\forall k$.
First, we also observe that the performance gap between the lower bound and the ergodic sum rate  is very small under different LI levels, regardless of the relay antenna number. Then it is seen
that Theorem 1 is a very close approximation to the proposed rate bound, especially when the number of relay antennas is very large. Thus Theorem 1 is a good predictor for the ergodic sum rate in the case of large antennas. However, the derived approximated results match not so well with the ergodic sum rates when not so many antennas are used, especially for the case of large LI, for example, the approximate result is 0.85 bits/s/Hz higher than the ergodic sum rate for $N_t=N_r=10$ and $\sigma_{LI}^2=10$ dB.
Furthermore, the high LI level decreases the sum rate significantly. In order to achieve the same sum rate in higher LI level, more antennas can be employed at the relay to suppress the LI.

\begin{figure}
\centering
\includegraphics[width=2.5in,height=1.9in]{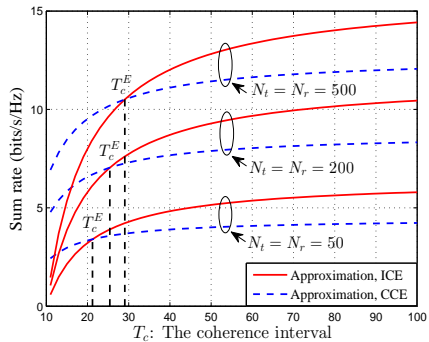}
\caption{Sum rate vs. $T_c$ in symmetric system with different number of relay antennas ($P_R = K P_S$).}\label{compairICEwithCCE-SRvsTc}
\end{figure}

\begin{figure}
\centering
\includegraphics[width=2.5in,height=1.9in]{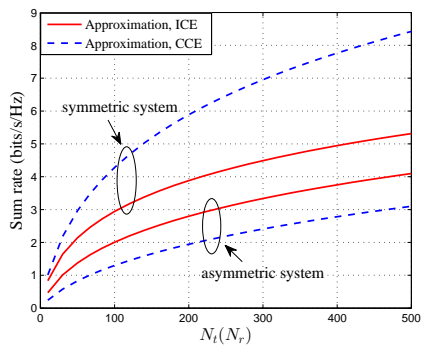}
\caption{Sum rate vs. $N_t$ in symmetric and asymmetric systems ($P_R = K P_S$, $T_c = 15$).}\label{compairICEwithCCE-SRvsNt}
\end{figure}
\subsection{Comparison between ICE and CCE}
In this subsection, the sum rate performance under ICE is compared with that under CCE. The path loss and shadow fading are taken into account with regard to the large-scale fading coefficients. Besides, it is reasonable to assume that $\mathbf{D}_\mathrm{u} = \mathbf{D}_\mathrm{d}$. Therefore, the large-scale fading coefficient $\beta_{ui}$ is given by
\begin{equation}\label{PLandShadowF1}
\beta_{ui} = \frac{10^{\omega_i/10}}{1+(d_i/d_0)^l},
\end{equation}
where $10^{\omega_i/10}$ is log-normally distributed with standard deviation of $\sigma$ dB, and $\omega_i\sim \mathcal{N}(0,\sigma^2)$ shows the log-normal attenuation which is also expressed in dB, $d_i$ denotes the distance from user $S_i$ to the relay and $d_0$ indicates the breakpoint in the path loss curve, and $l$ is the path loss exponent. In this paper, assuming that $d_i$ is uniformly distributed between 0 and 500 m, in addition, we set $\sigma = 8$ dB, $d_0 = 200$ m and $l=3.8$ \cite{massivemimo1}.

\begin{figure}
\centering
\includegraphics[width=2.5in,height=1.9in]{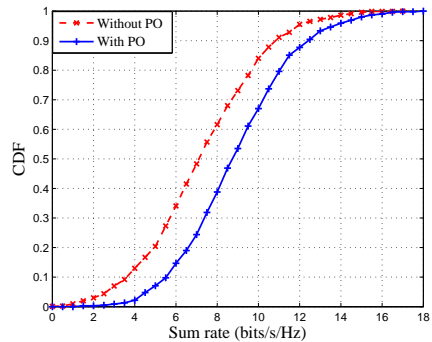}
\caption{CDF vs. sum rate ($N_t=N_r=200$).}\label{CDFvsSumRate}
\end{figure}

Fig. \ref{compairICEwithCCE-SRvsTc} depicts the sum rate versus the coherence interval in symmetric system and we set $P_R=KP_S$. According to Corollary \ref{cor2}, we have $T_c^E=21.2$ when $N_t=50$, $T_c^E=25.4$ when $N_t=200$ and $T_c^E=29.1$ when $N_t=500$. It is seen that the analytical results in Corollary \ref{cor2} match exactly with the numerical results. For instance, we observe that when $N_t=500$ and $T_c<29.1$, the CCE scheme outperforms the ICE scheme, while the ICE scheme performs better when $T_c>29.1$. In addition, the figure verifies that $T_c^E$ can be increased by adding the relay antenna number.

Moreover, we are also interested in the comparison under the case of asymmetric system. Here, the “asymmetric system” is defined as the case
where all $2K$ users are randomly located, i.e., all large-scale fading coefficients are generated based on (\ref{PLandShadowF1}). Thus, in Fig. \ref{compairICEwithCCE-SRvsNt}, we compare the system performance of the two channel estimation schemes in symmetric and asymmetric systems, respectively\footnote{Note that it is very difficult to derive the theoretical analysis as considering the asymmetric scenario. As a result, we only exhibit the simulation results of the asymmetric scenario.}.
In the figure, we set $T_c=15$ and $P_R=KP_S$. We observe that the sum rate performance in symmetric systems is better than that in asymmetric systems. Besides, the CCE scheme outperforms the ICE scheme in symmetric system, and the performance gap increases with the increase of $N_t$. However, the CCE scheme performs worse than ICE scheme in asymmetric system. Therefore, Fig.  \ref{compairICEwithCCE-SRvsNt} implies that the CCE scheme is preferable only in symmetric systems when $T_c<T_c^E$.

\subsection{Power Optimization (PO)}
In this subsection, we evaluate the power allocation algorithms. And we only consider the ICE scheme.

In Fig. \ref{CDFvsSumRate}, we examine the power allocation algorithm for sum rate maximization by comparing the cumulative distribution function (CDF) of sum rate in the case of power optimization to that without power optimization.
The practical large-scale fading model (\ref{PLandShadowF1}) is adopted and we choose $N_t=N_r=200$.
The ``Without PO'' curve corresponds to the uniform power allocation, i.e. $P_k=10$ dB ($\forall k$) and $P_R=KP_k$, and is obtained by performing (\ref{plofdiffuserP1}). The ``With PO'' curve is achieved by running the Algorithm 1\footnote{In our simulation, we stop the algorithm after a few iterations before it converges, and the maximum iteration number is set to be 5.}.
It is seen that the system achieves a large improvement for the sum rate performance with the optimal power allocation. For example, the 50th percentile of the sum rate is improved by about 1.9 bits/s/Hz.

\begin{figure}
\centering
\includegraphics[width=2.5in,height=1.9in]{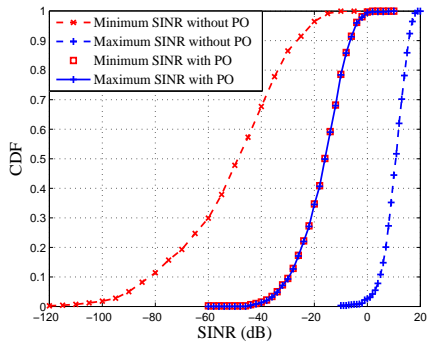}
\caption{CDF vs. SINR ($N_t=N_r=200$).}\label{CDFvsSINR}
\end{figure}

\begin{figure}
\centering
\includegraphics[width=2.5in,height=1.9in]{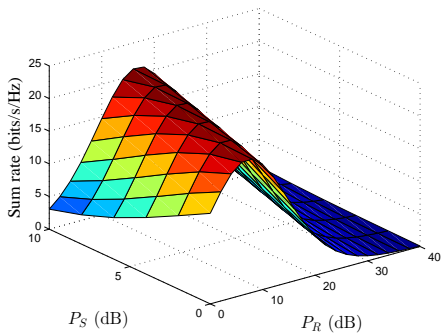}
\caption{Sum rate vs. $P_S$ and $P_R$ ($N_t=N_r=200$).}\label{SRvsPsPr-Powercontrol1}
\end{figure}

Then we evaluate the power allocation scheme based on the max-min fairness design criterion. Fig. \ref{CDFvsSINR} plots the CDF curves of the minimum and maximum SINR with and without power optimization, respectively. We observe that with optimal power allocation, the minimum SINR among all the users improves significantly. For example, the 50th percentile of minimum SINR with power optimization increases by about 31 dB  over that with uniform power allocation. Furthermore, it is seen that each user will achieve the same SINR under optimal power allocation.

Fig. \ref{SRvsPsPr-Powercontrol1} depicts the sum rate given by (\ref{eq19}) vs. $P_S$ and $P_R$ in the special case in which $\mathbf{D}_\mathrm{u}=\mathbf{D}_\mathrm{d}=\beta\mathbf{I}_{\mathrm{2K}}$, when $N_t=N_r=200$. Without loss of generality, we assume that $\beta=1$. We see that if $P_S$ is fixed, there will be an optimal $P_R$ for the maximum sum rate, and the performance will decline when $P_R$ is larger than this value, since the loop interference from the relay dominates the performance of FD systems in very large $P_R$.
In addition, for any finite $P_S$, the $P_R$ which contributes to the maximum sum rate has a linear relationship with $P_S$. Moreover, the figure indicates that the maximum sum rate is achieved when $P_S$ equals the peak value ($P_S^{max}=10$ dB). Therefore, Fig. \ref{SRvsPsPr-Powercontrol1} matches exactly with the result of Corollary \ref{corPowerControl1}, which gives the linear factor $\eta = 0.95K$ between $P_R$ and $P_S$.

\begin{figure}
\centering
\includegraphics[width=2.5in,height=1.9in]{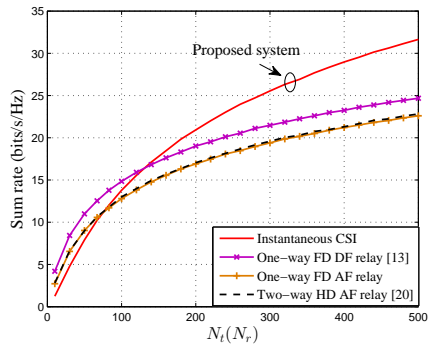}
\caption{Comparison of the proposed system with other schemes. (For fair comparison, in our scheme, $P_S{=}10$ dB, $P_R{=}KP_S{=}17$ dB; In one-way FD systems, $P_S{=}13$ dB, $P_R{=}17$ dB; In two-way HD systems, $P_S=13$ dB, $P_R{=}20$ dB. Note that the \textit{RF chains are conserved} in the comparison with HD systems \cite{massivemimo2}, i.e. the number of antennas at the HD relay node is $N_t$.)}\label{compareWithOneWayAndHD}
\end{figure}
\subsection{Comparison with Other Transmission Schemes}
In this subsection, we compare the sum rate performance of the system in this paper with those in one-way full-duplex systems and two-way half-duplex systems. We assume that $\mathbf{D}_\mathrm{u}=\mathbf{D}_\mathrm{d}=\mathbf{I}_{\mathrm{2K}}$ and only consider the ICE scheme.

In Fig. \ref{compareWithOneWayAndHD}, we compare the sum rate of the proposed system versus $N_t$ with those of other transmission schemes. The second and third curves in the legend denote the ergodic sum rates of the one-way FD decode-and-forward (DF) relaying system in \cite{massivemimo1} and the corresponding one-way AF relaying system, respectively, when the systems know the estimated instantaneous CSI. The fourth curve in the legend represents the ergodic sum rate for the two-way half-duplex AF relaying system in \cite{massivemimo2} with instantaneous CSI. For fair comparison, the total transmit powers on average during a time slot for different schemes are the same. And we consider the ``\textit{RF chain preserved}'' condition when compared with the half-duplex transmissions. Besides, all schemes employ MRC/MRT processing.

Fig. \ref{compareWithOneWayAndHD} indicates that our scheme performs significantly better than the performance of one-way FD relaying systems when $N_t$ is large, and the gain increases as $N_t$ grows. However, our scheme performs worse than the one-way system when $N_t$ is small (e.g. $N_t<130$ when compared with the one-way full-duplex DF relaying system in \cite{massivemimo1}), because there exist more interference terms in the proposed scheme than the one-way FD relaying (such as inter-user interferences, self-LI and more inter-pair interferences), and the interferences cannot be greatly reduced under small $N_t$. When the interference effect is greater than the additional multiplexing gain due to two-way relaying, the proposed scheme will provide a worse performance than the one-way relaying.
Besides, we observe that our scheme outperforms two-way HD relaying systems when $N_t$ is large, because the loop interference and inter-user interference in FD systems can be sharply decreased in large $N_t$ and FD systems can utilize time resources more efficiently. However, HD systems perform a little better in small $N_t$ (e.g. $N_t<75$), since the loop and inter-user interferences cannot be neglected in small $N_t$, thus leading to the performance degradation of FD systems.

\section{Conclusion}\label{conclusion}
In this paper, we investigated the achievable rate of a multi-pair two-way full-duplex AF relay system, where the relay adopted MRC/MRT processing and was equipped with large-scale antennas. When the number of relay antennas was large and finite, the approximate sum rates for the system were derived based on statistical channels, under individual channel estimation and composite channel estimation, respectively. It was shown that the derived sum rate expression is a tight approximation of the ergodic sum rate. In addition, we compared the two channel estimation schemes in terms of the achievable rate and showed that in symmetric systems, the composite channel estimation scheme performs better than the individual channel estimation scheme when the coherence interval is smaller than a certain value, and vice versa. Moreover, the power controls for the users and the relay were derived based on the achievable rate maximization and max-min fairness criterion, respectively. And the numerical results verified the accuracy of the analysis.

\appendices
\section*{appendix}
\subsection{Proofs of Equations (\ref{eq16}) and (\ref{eq17})}
Substituting (\ref{eq15}) into (\ref{eq5}) and (\ref{eq6}), and based on the property $\mathrm{Tr}(\mathbf{AB})=\mathrm{Tr}(\mathbf{BA})$, we obtain
\begin{align}
&\Delta_1 = \mathrm{Tr}\left\{ \mathbb{E}\left[ \left(\mathbf{\hat{F}}^H\mathbf{\hat{F}}\right)^T \mathbf{T} \left(\mathbf{\hat{G}}^H\mathbf{G} \mathbf{G}^H\mathbf{\hat{G}} \right) \mathbf{T} \right] \right\}, \label{eq1a}\\
&\Delta_2 = \mathrm{Tr}\left\{ \mathbb{E}\left[ \left(\mathbf{\hat{F}}^H\mathbf{\hat{F}}\right)^T \mathbf{T} \left(\mathbf{\hat{G}}^H \mathbf{\hat{G}} \right) \mathbf{T} \right] \right\}. \label{eq2a}
\end{align}

Assume that the matrices $\mathbf{F}$ and $\mathbf{G}$ are mutually independent, we can easily get that
\begin{align}
& \mathbb{E}\left\{\mathbf{\hat{F}}^H\mathbf{\hat{F}}\right\} =\mathbb{E}\left\{\left(\mathbf{F}+\mathbf{Z}_\mathrm{t}\right)^H
\left(\mathbf{F}+\mathbf{Z}_\mathrm{t}\right)\right\}
= N_t\mathbf{\hat{D}}_\mathrm{d}, \label{eq4a}\\
& \mathbb{E}\left\{\mathbf{\hat{G}}^H\mathbf{\hat{G}}\right\} =\mathbb{E}\left\{\left(\mathbf{G}+\mathbf{Z}_\mathrm{r}\right)^H
\left(\mathbf{G}+\mathbf{Z}_\mathrm{r}\right)\right\}
= N_r\mathbf{\hat{D}}_\mathrm{u}. \label{eq5a}
\end{align}
Substituting (\ref{eq4a}) and (\ref{eq5a}) into (\ref{eq2a}), we obtain (\ref{eq17}) for $\Delta_2$.

Besides, for obtaining $\Delta_1$, we calculate
\begin{align}\label{1delta1}
&\mathbb{E}\left\{\mathbf{\hat{G}}^H\mathbf{G} \mathbf{G}^H\mathbf{\hat{G}} \right\} = \mathbb{E}\left\{ \left( \mathbf{G} + \mathbf{Z}_\mathrm{r} \right)^H\mathbf{G} \mathbf{G}^H \left( \mathbf{G} + \mathbf{Z}_\mathrm{r} \right) \right\} \nonumber\\
&\hspace{9mm}=\mathbb{E}\left\{\mathbf{G}^H\mathbf{G} \mathbf{G}^H\mathbf{G} \right\}
+\mathbb{E}\left\{\mathbf{Z}_\mathrm{r}^H\mathbf{G} \mathbf{G}^H\mathbf{Z}_\mathrm{r} \right\} \nonumber\\
&\hspace{9mm}+\mathbb{E}\left\{\mathbf{G}^H\mathbf{G} \mathbf{G}^H\mathbf{Z}_\mathrm{r} \right\}
+ \mathbb{E}\left\{\mathbf{Z}_\mathrm{r}^H\mathbf{G} \mathbf{G}^H\mathbf{G} \right\}.
\end{align}
Since $\mathbb{E}\left\{\mathbf{\hat{F}}^H\mathbf{\hat{F}}\right\}$ is a diagonal matrix, we only need to know the diagonal entries of $\mathbb{E}\left\{\mathbf{\hat{G}}^H\mathbf{G} \mathbf{G}^H\mathbf{\hat{G}} \right\}$.

Firstly, we calculate $\mathbb{E}\left\{\mathbf{G}^H\mathbf{G} \mathbf{G}^H\mathbf{G} \right\}$. Let $\mathbf{M} = \mathbf{G}^H\mathbf{G}$, thus the $i$-th diagonal entry of $\mathbb{E}\left\{\mathbf{MM}^H\right\}$ is given by
\begin{equation}\label{delta1M1}
\mathbb{E}\left[ \mathbf{MM}^H \right]_{ii}
= \mathbb{E}\left\{ \left| \mathbf{g}_\mathrm{i}^H\mathbf{g}_\mathrm{i} \right|^2 \right\}
+\sum_{j=1,j\neq i}^{2K}{ \mathbb{E}\left\{ \left| \mathbf{g}_\mathrm{i}^H\mathbf{g}_\mathrm{j} \right|^2 \right\} }.
\end{equation}
Based on the properties of Wishart matrix \cite{randomMatrix}, we can get that
\begin{equation}\label{delta1M2}
\mathbb{E}\left\{ \left| \mathbf{g}_\mathrm{i}^H\mathbf{g}_\mathrm{i} \right|^2 \right\}
{=} \mathbb{E}\left\{ \mathrm{Tr}\left[\left( \mathbf{g}_\mathrm{i}^H\mathbf{g}_\mathrm{i} \right)^2\right] \right\}
{=}N_r\left( N_r+1 \right)\beta_{ui}^2.
\end{equation}
Besides,
\begin{equation}\label{delta1M3}
\mathbb{E}\left\{ \left| \mathbf{g}_\mathrm{i}^H\mathbf{g}_\mathrm{j} \right|^2 \right\}
{=}\mathbb{E}\left\{  \mathbf{g}_\mathrm{i}^H \mathbb{E}\left[\mathbf{g}_\mathrm{j}\mathbf{g}_\mathrm{j}^H\right]\mathbf{g}_\mathrm{i}  \right\}
{=}N_r\beta_{ui}\beta_{uj}, j{\neq} i.
\end{equation}
Thus we have
\begin{equation}\label{delta1M4}
\mathbb{E}\left[ \mathbf{G}^H\mathbf{G} \mathbf{G}^H\mathbf{G} \right]_{ii}
=N_r^2\beta_{ui}^2 + N_r\beta_{ui}\sum_{j=1}^{2K}\beta_{uj}.
\end{equation}

Secondly, we obtain
\begin{align}\label{2delta1}
\mathbb{E}\left\{\mathbf{Z}_\mathrm{r}^H\mathbf{G} \mathbf{G}^H\mathbf{Z}_\mathrm{r} \right\}
&= \mathbb{E}\left\{\mathbf{Z}_\mathrm{r}^H \mathrm{E}\left[ \mathbf{G} \mathbf{G}^H \right]\mathbf{Z}_\mathrm{r} \right\} \nonumber\\
&=\mathrm{Tr}(\mathbf{D}_\mathrm{u}) \cdot \mathbb{E}\left\{\mathbf{Z}_\mathrm{r}^H\mathbf{Z}_\mathrm{r} \right\} \nonumber\\
&=\mathrm{Tr}(\mathbf{D}_\mathrm{u}) \cdot N_r\frac{\sigma_{nr}^2}{\tau P_p} \mathbf{I}_{\mathrm{2K}}.
\end{align}

And then, we calculate $\mathbb{E}\left\{\mathbf{G}^H\mathbf{G} \mathbf{G}^H\mathbf{Z}_\mathrm{r} \right\}$. The $j$-th row of $\mathbb{E}\left\{\mathbf{G}^H\mathbf{G} \mathbf{G}^H \right\}$ is represented as
\begin{align}\label{3delta1}
&\mathbb{E}[ \mathbf{G}^H\mathbf{G} \mathbf{G}^H ]_\mathrm{j}
=\mathbb{E}\left\{\mathbf{g}_\mathrm{j}^H\mathbf{g}_\mathrm{j}\mathbf{g}_\mathrm{j}^H
\right\}
+\sum_{i=1,i\neq j}^{2K}{ \mathbb{E}\left\{\mathbf{g}_\mathrm{j}^H\mathbf{g}_\mathrm{i}\mathbf{g}_\mathrm{i}^H
\right\} } \nonumber\\
&\hspace{10mm}=\mathbb{E}\left\{\left\|\mathbf{g}_\mathrm{j}\right\|_2^2 \mathbf{g}_\mathrm{j}^H
\right\}
+\sum_{i=1,i\neq j}^{2K}{ \mathbb{E}\left\{\mathbf{g}_\mathrm{j}^H
\right\} \mathbb{E}\left\{\mathbf{g}_\mathrm{i}\mathbf{g}_\mathrm{i}^H
\right\} }.
\end{align}
It is apparent that $\mathbb{E}\{\mathbf{g}_\mathrm{j}^H\} = \mathbf{0}_{1\times\mathrm{N_r}} $ and $\mathbb{E}\{\mathbf{g}_\mathrm{i}\mathbf{g}_\mathrm{i}^H\} = \beta_{ui}\mathbf{I}_{\mathrm{N_r}} $, $\forall i, j$. And since $x^3f(x)$ ($f(x)$ is an even function which represents the probability density function for the real or imaginary part of the element of $\mathbf{g}_\mathrm{j}$) is an odd function, we can easily get that
\begin{equation}\label{7a}
\mathbb{E}\left\{ \| \mathbf{g}_\mathrm{j}\|_2^2 \mathbf{g}_\mathrm{j}^H \right\} = \mathbf{0}_{1\times \mathrm{N_r}}, ~\forall j.
\end{equation}
Thus we obtain $\mathbb{E}\left\{\mathbf{G}^H\mathbf{G} \mathbf{G}^H \right\} = \mathbf{0}_{\mathrm{2K}\times \mathrm{N_r}}$, then we get
\begin{equation}\label{4delta1}
\mathbb{E}\left\{\mathbf{G}^H\mathbf{G} \mathbf{G}^H\mathbf{Z}_\mathrm{r} \right\}
{=}\mathbb{E}\left\{\mathbf{G}^H\mathbf{G} \mathbf{G}^H\right\}
\mathbb{E}\left\{\mathbf{Z}_\mathrm{r} \right\}
{=}\mathbf{0}_{\mathrm{2K}\times \mathrm{2K}}.
\end{equation}

Based on (\ref{delta1M4}), (\ref{2delta1}) and (\ref{4delta1}), we have
\begin{align}\label{5delta1}
\mathbb{E}\left[ \mathbf{\hat{G}}^H\mathbf{G} \mathbf{G}^H\mathbf{\hat{G}} \right]_{ii}
&{=}N_r^2\beta_{ui}^2 {+} N_r\beta_{ui}\sum_{j=1}^{2K}\beta_{uj}
{+} N_r\frac{\sigma_{nr}^2}{\tau P_p}\mathrm{Tr}(\mathbf{D}_\mathrm{u}) \nonumber\\
&{=}N_r^2\beta_{ui}^2 {+} N_r\hat{\beta}_{ui}\sum_{j=1}^{2K}\beta_{uj}.
\end{align}
Substituting (\ref{eq4a}) and (\ref{5delta1}) into (\ref{eq1a}), we achieve $\Delta_1$ in (\ref{eq16}).

\subsection{Proof of Theorem 1}
From (\ref{eq14}), to obtain $\gamma_k$, we give the following calculations.
\subsubsection{Compute $\mathbb{E}\left\{ \mathbf{f}^T_\mathrm{k}\mathbf{W}\mathbf{g}_{\mathrm{k}'} \right\}$}
With MRC/MRT processing, we have
\begin{align}\label{eq1b}
\mathbb{E}\left\{ \mathbf{f}^T_\mathrm{k}\mathbf{W}\mathbf{g}_{\mathrm{k}'} \right\}
&= \mathbb{E}\left\{ \mathbf{f}^T_\mathrm{k} \left(\mathbf{F}+\mathbf{Z}_\mathrm{t}\right)^* \right\} \mathbf{T} \mathbb{E}\left\{ \left(\mathbf{G}+\mathbf{Z}_\mathrm{r}\right)^H\mathbf{g}_{\mathrm{k}'} \right\} \nonumber\\
&= (N_t\beta_{dk} \mathbf{1}_\mathrm{k}) \mathbf{T} (N_r\beta_{uk'} \mathbf{1}_\mathrm{k'}^T) \nonumber\\
&= N_tN_r\beta_{dk}\beta_{uk'},
\end{align}
where $\mathbf{1}_\mathrm{k}$ is a $1\times2K$ vector whose $k$-th element is 1, and the others are zeros.

\begin{figure*}[!t] % 加 [bp] 的话是放在页面底端！！！
% ensure that we have normalsize text
\normalsize
\setcounter{equation}{99}
\begin{equation}\label{exactSINRofLowerbound}
\gamma_k = \frac{N_t}
{
\begin{array}{l}
  \sum\limits_{j=1,j\neq k}^{2K} \left( \kappa \frac{\beta_{uj}\hat{\beta}_{uk'}}{\beta_{uk'}^2} + \frac{\hat{\beta}_{dj'}\beta_{uj}^2}{\beta_{dk}\beta_{uk'}^2} + \frac{\beta_{uj}}{N_r\beta_{dk}\beta_{uk'}^2} \sum\limits_{i=1}^{2K}\hat{\beta}_{di}\hat{\beta}_{ui'} \right)
  + \left( \frac{P_R\sigma_{LI}^2}{P_S} + \frac{\sigma_{nr}^2}{P_S} \right)\left(  \frac{\kappa\hat{\beta}_{uk'}}{\beta_{uk'}^2} + \frac{1}{N_r\beta_{dk}\beta_{uk'}^2} \sum\limits_{i=1}^{2K}\hat{\beta}_{di}\hat{\beta}_{ui'} \right)  \\
  + \frac{1}{\beta_{dk}^2\beta_{uk'}^2} \left( \sum\limits_{i\in \mathrm{U}_k} \sigma_{k,i}^2 + \frac{\sigma_n^2}{P_S} \right) \left[ \frac{P_S}{P_R}\sum\limits_{i=1}^{2K}\hat{\beta}_{di} \left( \beta_{ui'}^2 + \frac{\hat{\beta}_{ui'}}{N_r} \sum\limits_{j=1}^{2K}\beta_{uj} \right) + \frac{P_R\sigma_{LI}^2 + \sigma_{nr}^2}{N_rP_R}\sum\limits_{i=1}^{2K}\hat{\beta}_{di}\hat{\beta}_{ui'} \right]
\end{array}
}.
\end{equation}
\hrulefill
% The spacer can be tweaked to stop underfull vboxes.
\vspace*{0.0pt}
\end{figure*}

\setcounter{equation}{89}

\subsubsection{Compute $\mathbb{V}\mathrm{ar}( \mathbf{f}^T_\mathrm{k}\mathbf{W}\mathbf{g}_{\mathrm{k}'} )$}
It is calculated that
\begin{equation}\label{eq2b}
\mathbb{V}\mathrm{ar}( \mathbf{f}^T_\mathrm{k}\mathbf{W}\mathbf{g}_{\mathrm{k}'} )
= \mathrm{E}\left\{ \left| \mathbf{f}^T_\mathrm{k}\mathbf{W}\mathbf{g}_{\mathrm{k}'}
 \right|^2 \right\}
- \left|\mathbb{E}\left\{ \mathbf{f}^T_\mathrm{k}\mathbf{W}\mathbf{g}_{\mathrm{k}'} \right\}\right|^2,
\end{equation}
in which
\begin{align}\label{eq2bb}
&\mathrm{E}\left\{ \left| \mathbf{f}^T_\mathrm{k}\mathbf{W}\mathbf{g}_{\mathrm{k}'}
\right|^2 \right\}
=\mathrm{E}\left\{  \mathrm{Tr}\left[ \mathbf{f}^T_\mathrm{k}\mathbf{W}\mathbf{g}_{\mathrm{k}'}
\mathbf{g}_{\mathrm{k}'}^H \mathbf{W}^H\mathbf{f}_\mathrm{k}^* \right]  \right\} \nonumber\\
&\hspace{10mm}=\mathrm{Tr}\left\{
\mathbb{E}\left[\mathbf{\hat{F}}^T \mathbf{f}_\mathrm{k}^*\mathbf{f}_\mathrm{k}^T \mathbf{\hat{F}}^* \right] \mathbf{T}
\mathbb{E}\left[\mathbf{\hat{G}}^H \mathbf{g}_\mathrm{k'}\mathbf{g}_\mathrm{k'}^H \mathbf{\hat{G}} \right]
\mathbf{T}
\right\}.
\end{align}

First, we compute
\begin{align}\label{var1}
\mathbb{E}\left[\mathbf{\hat{F}}^T \mathbf{f}_\mathrm{k}^*\mathbf{f}_\mathrm{k}^T \mathbf{\hat{F}}^* \right]
&=\mathbb{E}\left[\mathbf{F}^T \mathbf{f}_\mathrm{k}^*\mathbf{f}_\mathrm{k}^T \mathbf{F}^* \right]
+\mathbb{E}\left[\mathbf{Z}_\mathrm{t}^T \mathbf{f}_\mathrm{k}^*\mathbf{f}_\mathrm{k}^T \mathbf{Z}_\mathrm{t}^* \right] \nonumber\\
&+\mathbb{E}\left[\mathbf{F}^T \mathbf{f}_\mathrm{k}^*\mathbf{f}_\mathrm{k}^T \mathbf{Z}_\mathrm{t}^* \right]
+\mathbb{E}\left[\mathbf{Z}_\mathrm{t}^T \mathbf{f}_\mathrm{k}^*\mathbf{f}_\mathrm{k}^T \mathbf{F}^* \right].
\end{align}
Based on (\ref{delta1M2}) and (\ref{7a}), the ($i$,$j$)-th element of $\mathbb{E}\left[\mathbf{F}^T \mathbf{f}_\mathrm{k}^*\mathbf{f}_\mathrm{k}^T \mathbf{F}^* \right]$ is
\begin{align}\label{var2}
\mathbb{E}\left[\mathbf{F}^T \mathbf{f}_\mathrm{k}^*\mathbf{f}_\mathrm{k}^T \mathbf{F}^* \right]_{ij}
&=
   \begin{cases}
     \mathbb{E}\left\{\left\| \mathbf{f}_\mathrm{k} \right\|_2^4\right\} \\
     \mathrm{E}\left\{ \mathbf{f}_\mathrm{i}^T \mathbb{E}\left[\mathbf{f}_\mathrm{k}^*\mathbf{f}_\mathrm{k}^T\right] \mathbf{f}_\mathrm{i}^* \right\}\\
     \mathrm{E}\left[ \mathbf{f}_\mathrm{i}^T \mathbf{f}_\mathrm{k}^*\mathbf{f}_\mathrm{k}^T \mathbf{f}_\mathrm{j}^* \right]
   \end{cases} \nonumber\\
&=
   \begin{cases}
     N_t(N_t+1)\beta_{dk}^2, & i=j=k; \\
     N_t\beta_{dk}\beta_{di}, & i=j\neq k; \\
     0, & i\neq j.
   \end{cases}
\end{align}
Similar to (\ref{4delta1}), we easily get $\mathbb{E}\left[\mathbf{F}^T \mathbf{f}_\mathrm{k}^*\mathbf{f}_\mathrm{k}^T \mathbf{Z}_\mathrm{t}^* \right] {=} \mathbf{0}_{\mathrm{2K}\times \mathrm{2K}}$. In addition,
\begin{equation}\label{var3}
\mathbb{E}\left[\mathbf{Z}_\mathrm{t}^T \mathbf{f}_\mathrm{k}^*\mathbf{f}_\mathrm{k}^T \mathbf{Z}_\mathrm{t}^* \right]
{=}\mathbb{E}\left\{\mathbf{Z}_\mathrm{t}^T \mathbb{E}\left[\mathbf{f}_\mathrm{k}^*\mathbf{f}_\mathrm{k}^T\right] \mathbf{Z}_\mathrm{t}^* \right\}
{=}N_t\beta_{dk}\frac{\sigma_{nr}^2}{\tau P_p} \mathbf{I}_{\mathrm{2K}}.
\end{equation}
Thus  $\mathbb{E}\left[\mathbf{\hat{F}}^T \mathbf{f}_\mathrm{k}^*\mathbf{f}_\mathrm{k}^T \mathbf{\hat{F}}^* \right]$ is a diagonal matrix and the ($i$, $i$)-th entry is
\begin{equation}\label{var4}
\mathbb{E}\left[\mathbf{\hat{F}}^T \mathbf{f}_\mathrm{k}^*\mathbf{f}_\mathrm{k}^T \mathbf{\hat{F}}^* \right]_{ii}
=
   \begin{cases}
     N_t^2\beta_{dk}^2 + N_t\beta_{dk}\hat{\beta}_{dk} & i=k; \\
     N_t\beta_{dk}\hat{\beta}_{di} & i\neq k.
   \end{cases}
\end{equation}

In the same way, we get that $\mathbb{E}\left[\mathbf{\hat{G}}^H \mathbf{g}_\mathrm{k'}\mathbf{g}_\mathrm{k'}^H \mathbf{\hat{G}} \right]$ is also a diagonal matrix, and the ($i$, $i$)-element is
\begin{equation}\label{var5}
\mathbb{E}\left[\mathbf{\hat{G}}^H \mathbf{g}_\mathrm{k'}\mathbf{g}_\mathrm{k'}^H \mathbf{\hat{G}} \right]_{ii}
{=}
   \begin{cases}
     N_r^2\beta_{uk'}^2 {+} N_r\beta_{uk'}\hat{\beta}_{uk'} & i=k'; \\
     N_r\beta_{uk'}\hat{\beta}_{ui} & i\neq k'.
   \end{cases}
\end{equation}

As a result, substituting (\ref{var4}) and (\ref{var5}) into (\ref{eq2bb}) and then substituting (\ref{eq2bb}) and (\ref{eq1b}) into (\ref{eq2b}), we obtain
\begin{align}\label{var6}
\mathbb{V}\mathrm{ar}( \mathbf{f}^T_\mathrm{k}\mathbf{W}\mathbf{g}_{\mathrm{k}'} )
&=N_t^2N_r\beta_{dk}^2\beta_{uk'}\hat{\beta}_{uk'}
+N_tN_r^2\beta_{dk}\hat{\beta}_{dk}\beta_{uk'}^2 \nonumber\\
&+N_tN_r\beta_{dk}\beta_{uk'}\sum_{i=1}^{2K}\hat{\beta}_{di}\hat{\beta}_{ui'}.
\end{align}

\subsubsection{Compute $\mathbb{E}\left\{\left|\mathbf{f}^T_\mathrm{k}\mathbf{W}\mathbf{g}_{\mathrm{j}}\right|^2 \right\}$, ($j\neq k, k'$)}
With the similar way as (\ref{eq2bb}), we get
\begin{align}\label{Egj1}
\mathbb{E}\left\{\left|\mathbf{f}^T_\mathrm{k}\mathbf{W}\mathbf{g}_{\mathrm{j}}\right|^2 \right\}
&=N_t^2N_r\beta_{dk}^2\beta_{uj}\hat{\beta}_{uk'}
+N_tN_r^2\beta_{dk}\hat{\beta}_{dj'}\beta_{uj}^2 \nonumber\\
&+N_tN_r\beta_{dk}\beta_{uj}\sum_{i=1}^{2K}{ \hat{\beta}_{di}\hat{\beta}_{di'} }.
\end{align}

\subsubsection{Compute $\mathbb{E}\left\{ \left\|\mathbf{f}^T_\mathrm{k}\mathbf{W}\right\|_2^2 \right\}$}
With $\mathbf{W}{=}\mathbf{\hat{F}}^*\mathbf{T}\mathbf{\hat{G}}^H$, based on (\ref{var4}), we have
\begin{align}\label{eq6b}
&\mathbb{E}\left\{ \left\|\mathbf{f}^T_\mathrm{k}\mathbf{W}\right\|_2^2 \right\}
= \mathbb{E}\left\{ \mathbf{f}^T_\mathrm{k} \mathbf{\hat{F}}^*\mathbf{T} \mathbb{E}\left[ \mathbf{\hat{G}}^H \mathbf{\hat{G}} \right] \mathbf{T}\mathbf{\hat{F}}^T \mathbf{f}^*_\mathrm{k}  \right\} \nonumber\\
&\hspace{15mm}= N_r \mathrm{Tr}\left\{ \mathbb{E}\left[ \mathbf{\hat{F}}^T \mathbf{f}^*_\mathrm{k}\mathbf{f}^T_\mathrm{k} \mathbf{\hat{F}}^*\right] \mathbf{T} \mathbf{\hat{D}}_\mathrm{u} \mathbf{T}  \right\} \nonumber\\
&\hspace{15mm}= N_r \sum_{i=1}^{2K}{ \left\{ \hat{\beta}_{ui'}\mathbb{E}\left[ \mathbf{\hat{F}}^T \mathbf{f}^*_\mathrm{k}\mathbf{f}^T_\mathrm{k} \mathbf{\hat{F}}^* \right]_{ii} \right\} } \nonumber\\
&\hspace{15mm}=N_t^2N_r\beta_{dk}^2\hat{\beta}_{uk'} + N_tN_r\beta_{dk}\sum_{i=1}^{2K}{ \hat{\beta}_{di}\hat{\beta}_{ui'} }.
\end{align}

So far, we can obtain $\gamma_k$  given by (\ref{exactSINRofLowerbound}) (see top of this page) by substituting $\alpha$, (\ref{eq1b}), (\ref{var6}), (\ref{Egj1}) and (\ref{eq6b}) into (\ref{eq14}). Then the closed-form expression for the proposed lower bound  given by (\ref{eq13}) can be achieved.

However, the expression (\ref{exactSINRofLowerbound}) is very lengthy. When $\kappa$ is fixed and $N_r\gg 2K$, we only retain the items with the highest power of $N_tN_r$ in equations (\ref{var6}) $\sim$ (\ref{eq6b}) for getting a more concise expression.
This approximation is also imposed on the factor $\alpha$. And then we arrive at (\ref{eq19}).
It is expected to be a tight approximation, especially in the regime of very large $N_r$.

\bibliography{reference}

% Generated by IEEEtran.bst, version: 1.13 (2008/09/30)
\begin{thebibliography}{10}
\providecommand{\url}[1]{#1}
\csname url@samestyle\endcsname
\providecommand{\newblock}{\relax}
\providecommand{\bibinfo}[2]{#2}
\providecommand{\BIBentrySTDinterwordspacing}{\spaceskip=0pt\relax}
\providecommand{\BIBentryALTinterwordstretchfactor}{4}
\providecommand{\BIBentryALTinterwordspacing}{\spaceskip=\fontdimen2\font plus
\BIBentryALTinterwordstretchfactor\fontdimen3\font minus
  \fontdimen4\font\relax}
\providecommand{\BIBforeignlanguage}[2]{{%
\expandafter\ifx\csname l@#1\endcsname\relax
\typeout{** WARNING: IEEEtran.bst: No hyphenation pattern has been}%
\typeout{** loaded for the language `#1'. Using the pattern for}%
\typeout{** the default language instead.}%
\else
\language=\csname l@#1\endcsname
\fi
#2}}
\providecommand{\BIBdecl}{\relax}
\BIBdecl

\bibitem{massiveMIMOnextg}
E.~G. Larsson, O.~Edfors, F.~Tufvesson, and T.~L. Marzetta, ``{Massive MIMO for
  next generation wireless systems},'' \emph{IEEE Commun. Mag.}, vol.~52,
  no.~2, pp. 186--195, February 2014.

\bibitem{massiveMIMOhowManyAntennas}
J.~Hoydis, S.~ten Brink, and M.~Debbah, ``{Massive MIMO in the UL/DL of
  Cellular Networks: How Many Antennas Do We Need?}'' \emph{IEEE J. Sel. Areas
  Commun.}, vol.~31, no.~2, pp. 160--171, February 2013.

\bibitem{massiveMIMO3GPPmeeting}
3GPP meetings for group R1. [Online]. Available:
  http://www.3gpp.org/dynareport/Meetings-R1.htm?Itemid=404.

\bibitem{fdsystems}
A.~C. Cirik, Y.~Rong, and Y.~Hua, ``{Achievable Rates of Full-Duplex MIMO
  Radios in Fast Fading Channels With Imperfect Channel Estimation},''
  \emph{IEEE Trans. Signal Process.}, vol.~62, no.~15, pp. 3874--3886, Aug.
  2014.

\bibitem{LIC10Stanford}
J.~I. Choi, M.~Jain, K.~Srinivasan, P.~Levis, and S.~Katti, ``Achieving single
  channel, full duplex wireless communication,'' in \emph{Proc. MobiCom}.\hskip
  1em plus 0.5em minus 0.4em\relax ACM, 2010, pp. 1--12.

\bibitem{LIC11Stanford}
M.~Jain, J.~I. Choi, T.~Kim, D.~Bharadia, S.~Seth, K.~Srinivasan, P.~Levis,
  S.~Katti, and P.~Sinha, ``Practical, real-time, full duplex wireless,'' in
  \emph{Proc. MobiCom}.\hskip 1em plus 0.5em minus 0.4em\relax ACM, 2011, pp.
  301--312.

\bibitem{MitigationofLIfdMIMO}
T.~Riihonen, S.~Werner, and R.~Wichman, ``{Mitigation of Loopback
  Self-Interference in Full-Duplex MIMO Relays},'' \emph{IEEE Trans. Signal
  Process.}, vol.~59, no.~12, pp. 5983--5993, Dec. 2011.

\bibitem{linearFDMIMORelay14C}
C.~Y.~A. Shang, P.~J. Smith, G.~K. Woodward, and H.~A. Suraweera, ``Linear
  transceivers for full duplex mimo relays,'' in \emph{Proc. Australian Commun.
  Theory Workshop (AusCTW)}, Feb. 2014, pp. 11--16.

\bibitem{LCdesignFDMIMO14TWC}
H.~A. Suraweera, I.~Krikidis, G.~Zheng, C.~Yuen, and P.~J. Smit,
  ``{Low-Complexity End-to-End Performance Optimization in MIMO Full-Duplex
  Relay Systems},'' \emph{IEEE Trans. Wireless Commun.}, vol.~13, no.~2, pp.
  913--927, Feb. 2014.

\bibitem{Experiment12TWC}
M.~Duarte, C.~Dick, and A.~Sabharwal, ``{Experiment-Driven Characterization of
  Full-Duplex Wireless Systems},'' \emph{IEEE Trans. Wireless Commun.},
  vol.~11, no.~12, pp. 4296--4307, December 2012.

\bibitem{FDcancel110dB}
D.~Bharadia, E.~McMilin, and S.~Katti, ``Full duplex radios,'' \emph{ACM
  SIGCOMM}, vol.~43, no.~4, pp. 375--386, Aug. 2013.

\bibitem{FDinRelayand100dBsuppress}
M.~Heino and \textit{et al.}, ``{Recent advances in antenna design and
  interference cancellation algorithms for in-band full duplex relays},''
  \emph{IEEE Commun. Mag.}, vol.~53, no.~5, pp. 91--101, May 2015.

\bibitem{massivemimo1}
H.~Q. Ngo, H.~A. Suraweera, M.~Matthaiou, and E.~G. Larsson, ``{Multipair
  Full-Duplex Relaying With Massive Arrays and Linear Processing},'' \emph{IEEE
  J. Sel. Areas Commun.}, vol.~32, no.~9, pp. 1721--1737, Sept. 2014.

\bibitem{inBandFD14JSAC}
A.~Sabharwal and \textit{et al.}, ``{In-Band Full-Duplex Wireless: Challenges
  and Opportunities},'' \emph{IEEE J. Sel. Areas Commun.}, vol.~32, no.~9, pp.
  1637--1652, Sept. 2014.

\bibitem{FDrelayin3GPP}
{\textit{Physical Layer for Relaying Operation (Release 10)}, LTE spec, 3GPP TS
  36.216}, June 2011.

\bibitem{ManavTWAFOSTBC}
A.~M. K. and M.~R. Bhatnagar, ``{Performance Analysis of Two-Way AF MIMO
  Relaying of OSTBCs With Imperfect Channel Gains},'' \emph{IEEE Trans. Veh.
  Technol.}, vol.~63, no.~8, pp. 4118--4124, Oct. 2014.

\bibitem{ManavTWSatelliteRelay}
M.~R. Bhatnagar, ``{Making Two-Way Satellite Relaying Feasible: A Differential
  Modulation Based Approach},'' \emph{IEEE Trans. Commun.}, vol.~63, no.~8, pp.
  2836--2847, Aug. 2015.

\bibitem{fdMIMOrelay}
G.~Zheng, ``{Joint Beamforming Optimization and Power Control for Full-Duplex
  MIMO Two-Way Relay Channel},'' \emph{IEEE Trans. Signal Process.}, vol.~63,
  no.~3, pp. 555--566, Feb. 2015.

\bibitem{HASura13ICC}
H.~A. Suraweera, H.~Q. Ngo, T.~Q. Duong, C.~Yuen, and E.~G. Larsson,
  ``Multi-pair amplify-and-forward relaying with very large antenna arrays,''
  in \emph{Proc. IEEE Int. Conf. Commun. (ICC)}, June 2013, pp. 4635--4640.

\bibitem{massivemimo2}
H.~Cui, L.~Song, and B.~Jiao, ``{Multi-Pair Two-Way Amplify-and-Forward
  Relaying with Very Large Number of Relay Antennas},'' \emph{IEEE Trans.
  Wireless Commun.}, vol.~13, no.~5, pp. 2636--2645, May 2014.

\bibitem{ergodicRate}
S.~Jin, X.~Liang, K.~K. Wong, X.~Gao, and Q.~Zhu, ``{Ergodic Rate Analysis for
  Multipair Massive MIMO Two-Way Relay Networks},'' \emph{IEEE Trans. Wireless
  Commun.}, vol.~14, no.~3, pp. 1480--1491, March 2015.

\bibitem{zhengzhengxiang}
Z.~Xiang, M.~Tao, and X.~Wang, ``{Massive MIMO Multicasting in Noncooperative
  Cellular Networks},'' \emph{IEEE J. Sel. Areas Commun.}, vol.~32, no.~6, pp.
  1180--1193, June 2014.

\bibitem{TWC09oneFDantenna}
H.~Ju, E.~Oh, and D.~Hong, ``{Improving efficiency of resource usage in two-hop
  full duplex relay systems based on resource sharing and interference
  cancellation},'' \emph{IEEE Trans. Wireless Commun.}, vol.~8, no.~8, pp.
  3933--3938, August 2009.

\bibitem{SPL12oneFDantenna}
Y.~Hua, P.~Liang, Y.~Ma, A.~C. Cirik, and Q.~Gao, ``{A Method for Broadband
  Full-Duplex MIMO Radio},'' \emph{IEEE Signal Process. Lett.}, vol.~19,
  no.~12, pp. 793--796, Dec. 2012.

\bibitem{HongyuCuioneFDantenna}
H.~Cui, M.~Ma, L.~Song, and B.~Jiao, ``{Relay Selection for Two-Way Full Duplex
  Relay Networks With Amplify-and-Forward Protocol},'' \emph{IEEE Trans.
  Wireless Commun.}, vol.~13, no.~7, pp. 3768--3777, July 2014.

\bibitem{globecom05twoFDantenna}
C.~K. Lo, S.~Vishwanath, and R.~W. Heath, ``Rate bounds for mimo relay channels
  using precoding,'' in \emph{Proc. IEEE GLOBECOM}, Nov. 2005, pp. 1172--1176.

\bibitem{li1}
Y.~Sung, J.~Ahn, B.~V. Nguyen, and K.~Kim, ``{Loop-interference suppression
  strategies using antenna selection in full-duplex MIMO relays},'' in
  \emph{Proc. Int. Symp. Intelligent Signal Process. and Commun. Syst.
  (ISPACS)}, Dec. 2011, pp. 1--4.

\bibitem{LImodelJSAC}
L.~J. Rodr¨ªguez, N.~H. Tran, and T.~Le-Ngoc, ``{Performance of Full-Duplex AF
  Relaying in the Presence of Residual Self-Interference},'' \emph{IEEE J. Sel.
  Areas in Commun.}, vol.~32, no.~9, pp. 1752--1764, Sept. 2014.

\bibitem{LInoiseModelAndLICmethodsurvey}
J.~Lee and T.~Q.~S. Quek, ``{Hybrid Full-/Half-Duplex System Analysis in
  Heterogeneous Wireless Networks},'' \emph{IEEE Trans. Wireless Commun.},
  vol.~14, no.~5, pp. 2883--2895, May 2015.

\bibitem{GaussianNoiseWorstCase}
I.~Shomorony and A.~S. Avestimehr, ``Is gaussian noise the worst-case additive
  noise in wireless networks?'' in \emph{Proc. IEEE Int. Symp. Inf. Theory},
  July 2012, pp. 214--218.

\bibitem{InterUserInterferenceModel}
B.~Yin, M.~Wu, C.~Studer, J.~R. Cavallaro, and J.~Lilleberg, ``Full-duplex in
  large-scale wireless systems,'' in \emph{Proc. Asilomar Conf. Signals, Syst.,
  Comput. (ASILOMAR)}, Nov. 2013, pp. 1623--1627.

\bibitem{relayDelayTWC2011}
T.~Riihonen, S.~Werner, and R.~Wichman, ``{Hybrid Full-Duplex/Half-Duplex
  Relaying with Transmit Power Adaptation},'' \emph{IEEE Trans. Wireless
  Commun.}, vol.~10, no.~9, pp. 3074--3085, September 2011.

\bibitem{relayDelayTWC2013}
G.~Zheng, I.~Krikidis, and B.~o.~Ottersten, ``{Full-Duplex Cooperative
  Cognitive Radio with Transmit Imperfections},'' \emph{IEEE Trans. Wireless
  Commun.}, vol.~12, no.~5, pp. 2498--2511, May 2013.

\bibitem{fixedGainRelay}
M.~O. Hasna and M.~S. Alouini, ``{A performance study of dual-hop transmissions
  with fixed gain relays},'' \emph{IEEE Trans. Wireless Commun.}, vol.~3,
  no.~6, pp. 1963--1968, Nov. 2004.

\bibitem{LSchannelEstimation}
M.~Biguesh and A.~B. Gershman, ``{Training-based MIMO channel estimation: a
  study of estimator tradeoffs and optimal training signals},'' \emph{IEEE
  Trans. Signal Process.}, vol.~54, no.~3, pp. 884--893, March 2006.

\bibitem{howMuchTraining}
B.~Hassibi and B.~M. Hochwald, ``{How Much Training is Needed in
  Multiple-Antenna Wireless Links?}'' \emph{IEEE Trans. Inf. Theory}, vol.~49,
  no.~4, pp. 951--963, April 2003.

\bibitem{lowerBoundR11TWC}
J.~Jose, A.~Ashikhmin, T.~L. Marzetta, and S.~Vishwanath, ``{Pilot
  Contamination and Precoding in Multi-Cell TDD Systems},'' \emph{IEEE Trans.
  Wireless Commun.}, vol.~10, no.~8, pp. 2640--2651, August 2011.

\bibitem{lowerBoundR15TIT}
E.~Bjornson, J.~Hoydis, M.~Kountouris, and M.~Debbah, ``{Massive MIMO Systems
  With Non-Ideal Hardware: Energy Efficiency, Estimation, and Capacity
  Limits},'' \emph{IEEE Trans. Inf. Theory}, vol.~60, no.~11, pp. 7112--7139,
  Nov. 2014.

\bibitem{lowerBoundR15TWC}
E.~Bjornson, M.~Matthaiou, and M.~Debbah, ``{Massive MIMO with Non-Ideal
  Arbitrary Arrays: Hardware Scaling Laws and Circuit-Aware Design},''
  \emph{IEEE Trans. Wireless Commun.}, vol.~14, no.~8, pp. 4353--4368, Aug.
  2015.

\bibitem{Xiaojunzheng}
X.~Zheng, E.~Liu, Z.~Zhang, X.~Qu, R.~Wang, X.~Yin, and F.~Liu, ``An efficient
  pilot scheme in large-scale two-way relay systems,'' \emph{IEEE Commun.
  Lett.}, vol.~19, no.~6, pp. 1061--1064, June 2015.

\bibitem{approximationOnePlusR}
P.~C. Weeraddana, M.~Codreanu, M.~Latva-aho, and A.~Ephremides, ``{Resource
  Allocation for Cross-Layer Utility Maximization in Wireless Networks},''
  \emph{IEEE Trans. Veh. Technol.}, vol.~60, no.~6, pp. 2790--2809, July 2011.

\bibitem{cvx}
M.~Grant and S.~Boyd, ``{CVX}: Matlab software for disciplined convex
  programming, version 2.1,'' \url{http://cvxr.com/cvx}, Mar. 2014.

\bibitem{randomMatrix}
A.~M. Tulino and S.~Verd{\'u}, ``{Random matrix theory and wireless
  communications},'' \emph{Communications and Inf. theory}, vol.~1, no.~1, pp.
  1--182, Jun. 2004.

\end{thebibliography}
\end{document}